\documentclass{article}
\usepackage{appendix}
\usepackage[T1]{fontenc}
\usepackage{amsfonts,mathrsfs}
\usepackage[english]{babel}
\usepackage{graphicx,multirow}
\usepackage{indentfirst}
\usepackage[colorlinks=true]{hyperref}
\usepackage{amsmath, amsthm, amssymb}
\usepackage{hyperref}
\usepackage{subcaption}
\usepackage{pdfpages}
\usepackage{lipsum}
\usepackage[left=2cm, right=2cm, bottom=3cm]{geometry}
\usepackage{hyperref}
\usepackage{xcolor}
\usepackage{tcolorbox}
\tcbuselibrary{skins}
\usepackage{colortbl}
\usepackage{bbm}
\usepackage{cite}
\usepackage{dsfont}
\usepackage[resetlabels,labeled]{multibib}
\usepackage{epsf,epsfig}
\usepackage{multirow}
\usepackage{caption}
\usepackage{bbm, mathrsfs}
\usepackage{titlesec}
\usepackage[utf8]{inputenc}
\usepackage{float}
\usepackage{geometry}
\usepackage{graphicx}
\usepackage{listings}
\usepackage{textcomp}
\usepackage{booktabs}
\usepackage[colorlinks=true]{hyperref}
\usepackage{xcolor}
\usepackage{float,color}
\usepackage{array}
\usepackage[normalem]{ulem}
\usepackage{tcolorbox}

\hypersetup{linkcolor=blue}
\definecolor{mygreen}{RGB}{28,172,0} 
\definecolor{mylilas}{RGB}{170,55,241}
\newfloat{figure}{H}{lof}
\newfloat{table}{H}{lot}
\floatname{figure}{\figurename}
\floatname{table}{\tablename}
\definecolor{mycolor}{RGB}{0, 204, 204}
\newtheorem{theorem}{Theorem}[section]

\newtheorem{assumption}[theorem]{Assumption}

\newtheorem{lemma}[theorem]{Lemma}

\newtheorem{remark}[theorem]{Remark}

\numberwithin{equation}{section}

\renewcommand{\eqref}[1]{(\ref{#1})}

\definecolor{mycolor}{RGB}{0, 204, 204}
\definecolor{mycolor1}{RGB}{   232, 247, 251 }
\definecolor{mycolor2}{RGB}{   1, 104, 137 }
\newtcolorbox{mybox}[1]{colback=red!5!white,colframe=red!75!black,fonttitle=\bfseries,title=#1}
\newtcolorbox{resultbox}[1]{colback=red!5!white,colframe=red!75!black,fonttitle=\bfseries,title=#1}
\newtcolorbox{applicationbox}[1]{colback=mycolor1,colframe=mycolor2,fonttitle=\bfseries,title=#1}
\newtcolorbox{citationbox}[1]{colback=white,colframe=white,fonttitle=\bfseries,title=#1}

\begin{document}
	\title{\textbf{The Effective Reproduction Number in the Kermack-McKendrick model with age of infection and reinfection}}
	\author{\textsc{JIAYI LI$^{(a),1}$, ZHIHUA LIU$^{(a),}$\footnote{Corresponding author. e-mail: zhihualiu@bnu.edu.cn.} ,ZIHAN WANG$^{(a),2}$}\\
		{\small \textit{$^{(a)}$School of Mathematical Sciences, Beijing Normal University,}} \\
		{\small \textit{Beijing 100875, People’s Republic of China.}} 
	}

	\maketitle
	
	\begin{abstract}
		This study introduces a novel epidemiological model that expands upon the Kermack-McKendrick model by incorporating the age of infection and reinfection. By including infection age, we can classify participants, which enables a more targeted analysis within the modeling framework. The reinfection term addresses the real-world occurrences of secondary or recurrent viral infections. In the theoretical part, we apply the contraction mapping principle, the dominated convergence theorem, and the properties of Volterra integral equations to derive analytical expressions for the number of newly infected individuals denoted by $N(t)$. Then, we establish a Volterra integral equation for $N(t)$ and study its initial conditions for both a single cohort and multiple cohorts. From this equation, we derive a method for identifying the effective reproduction number, denoted as $\mathcal{R}(t)$. In the practical aspect, we present two distinct methods and separately apply them to analyze the daily new infection cases from the 2003 SARS outbreak in Singapore and the cumulative number of deaths from the COVID-19 epidemic in China. This work effectively bridges theoretical epidemiology and computational modeling, providing a robust framework for analyzing infection dynamics influenced by infection-age-structured transmission and reinfection mechanisms.
	\end{abstract}
	\medskip 
	
	
	\section{Introduction}

Following the 2002 outbreak of the Severe Acute Respiratory Syndrome Coronavirus (SARS-CoV-1) \cite{ref35} and the 2012 outbreak of the Middle East Respiratory Syndrome Coronavirus (MERS-CoV) \cite{ref36}, SARS-CoV-2, which causes COVID-19, marks the third coronavirus pandemic in two decades. Human-endemic coronaviruses, including types 229E, OC43, NL63, and HKU1, are among the primary pathogens responsible for the common cold \cite{ref37}. Unlike the human-endemic coronaviruses, SARS-CoV-2 has caused severe respiratory illness and significant mortality, particularly among older adults and those with underlying health conditions. SARS-CoV-2 spreads through respiratory droplets, contact transmission, and surfaces that have been contaminated. Clinical symptoms such as coughing and sputum production, with the viral load in infected individuals, are crucial factors influencing transmissibility, which can vary significantly within days throughout the infectious period. The existence of an individual infection period of variable length and magnitude among infected individuals is a proven fact in all contagious diseases. This is the reason why we introduce the age of infection in the following. Current directions on infectious diseases include the levels of population \cite{ref15,ref16,ref29}, individual \cite{ref2}, and viral \cite{ref17}. 

One of the fundamental epidemiological parameters is the basic reproduction number $\mathcal{R}_0$. This is defined as the average number of secondary infections caused by an infected individual in a fully susceptible population throughout his or her infectious period. In contrast to the classical methods in \cite{ref25} for calculating $\mathcal{R}_0$, \cite{ref2} uses the daily reproduction number $\mathcal{R}_0(a)$ to calculate $\mathcal{R}_0$ and shows that $\mathcal{R}_0(a)$ can be obtained by solving a Volterra integral equation that depends on the flow of new infected individuals considering a single cohort of individuals under the assumptions of a constant transmission rate and a fixed susceptible population. Considering factors such as herd immunity and quarantine measures, researchers avoid constant assumptions about the transmission rate and the number of susceptible individuals, and pay attention to the effective reproduction number $\mathcal{R}(t)$ in \cite{ref46,ref47}. This also gives rise to the problem of identifying the parameter $\mathcal{R}(t)$ in models incorporating the age of infection.

This paper proposes a method for identifying the effective reproduction number. By this method, given the number of newly infected individuals $N(t)$, the reproductive power $\mathcal{R}(t, a)$ (i.e., the rate at which an infectious individual at calendar time $t$ and infection-age $a$ produces secondary cases) can be derived in our model. Then, we can calculate the effective reproduction number $\mathcal{R}(t)$ by using the reproductive power $\mathcal{R}(t, a)$. In our model, there are two critical components in the infectious disease transmission process related to the age of infection: the initial distribution $i_0(a)$ and the probability of being infectious $\beta(a)$. Concerning the initial distribution $i_0(a)$, we use a Dirac mass as the initial distribution of our model and demonstrate the justification for using this extended initial condition. Recently, there has been a growing interest in measure-valued solutions for structured population models, with relevant assumptions regarding initial distributions discussed in \cite{ref2} and \cite{ref23}. Concerning the probability of infectiousness $\beta(a)$, we continue to use the assumptions of \(\beta(a)\) from \cite{ref2}. 

The Kermack-McKendrick model, formulated by William Kermack and Anderson McKendrick in 1927 \cite{ref43}, established the basis for compartmental models. In \cite{ref45}, a first-order hyperbolic partial differential equation is used to describe the McKendrick model known as the age-structured model. Current research \cite{ref2} explores the parameter identifiability of a Kermack-McKendrick model that incorporates the age of infection, primarily organized into $S$ and $I$ compartments. Reference \cite{ref25} studies a similar infection-age-structured model but with a different compartmental configuration of $(S, I, R)$. Building on these previous works, we introduce a reinfection term to extend the existing infection-age-structured model framework.

	This paper introduces an innovative epidemiological model based on the classical Susceptible-Infected-Recovered-Susceptible (SIRS) framework. The model shows a more detailed analysis of disease transmission patterns. 
    In Section \ref{Section2}, we first give the definitions of parameters and then present the model through two approaches: the formulation of partial differential equations and Volterra integral equations. 
    Section \ref{Section3} studies the model under two initial conditions: a single cohort and multiple cohorts of infected patients. After that, we define the reproductive power $\mathcal{R}(t,a)$ and establish a connection between the number of newly infected individuals $N(t)$ and the reproductive power $\mathcal{R}(t,a)$. Given the reproductive power $\mathcal{R}(t,a)$, we can derive the effective reproduction number $\mathcal{R}(t)$. Section \ref{Section4} uses the cumulative number of deaths from existing data to estimate the number of newly infected individuals $N(t)$, infected individuals, susceptible individuals, and more. This section also provides an estimation approach for data with uncertain numbers of newly infected individuals. Section \ref{Section5} is the core of the parameter identification methodology. It begins by discretizing the model and then proceeds to study two situations, providing an approach to identify the effective reproduction number $\mathcal{R}(t)$ by the number of newly infected individuals $N(t)$ in the discrete model.
    Section \ref{Section6} illustrates the application of the methodology proposed in Section \ref{Section4} and Section \ref{Section5}. In this section, we apply the 2003 SARS outbreak data of the newly infected individuals in Singapore to the parameter identification approach developed in Section \ref{Section5}. Then, we discuss the identification of the reproductive power $\mathcal{R}(t,a)$ using the cumulative number of deaths of the COVID-19 epidemic in China, where the data of the new infected individuals cannot be determined on certain days.
	
	\section{Kermack-McKendrick model with age of infection and reinfection}
	
	\label{Section2} 
    The infectiousness of infected individuals is known to depend on the time since the individual was infected, called the age of infection. According to \cite{ref1,ref29}, $i(t, a)$ represents the number of infected individuals with respect to infection age $a$ at time $t$. The integral 
	\begin{equation*}
		\int_{a_{1}}^{a_{2}}i(t,a)da,
	\end{equation*}
	means the number of infected at time $t$ with the age of infection between $a_{1}$ and $a_{2}$. Therefore, the total number of infected individuals at time $t$ is
	
	\begin{equation*}
		I(t)=\int_{0}^{\infty}i(t,a)da.  \tag{2.1}  \label{eq:2.1}
	\end{equation*}
	Let $\beta(a)$ be the probability of being infectious at the age of infection $a$. Then the
	total number of infectious individuals at
	time $t$ is 
	\begin{equation*}
		C(t)=\int_{0}^{\infty}\beta(a)i(t,a)da.
	\end{equation*}

	\subsection*{(a) Partial differential equation formulation of the
		model}
	
    The widely known form of the Kermack-McKendrick model in \cite{ref43,ref3} is governed by ordinary differential equations, while its original version in 1927 incorporated assumptions structured by age of infection. We continue works on \cite{ref1,ref2,ref25} and propose the Kermack-McKendrick model with reinfection and age of infection below.
	\begin{align*}
		\left\{ 
		\begin{array}{l}
			S^{\prime}(t) =-\tau(t)S(t)\int_{0}^{\infty}\beta(a)i(t,a)da+\delta R(t), \\ 
			\partial_{t}i+\partial_{a}i =-(\nu+D)i(t,a), \\ 
			R^{\prime}(t) =\nu\int_{0}^{\infty}i(t,a)da-\delta R(t), \\ 
			i(t,0) =\tau(t)S(t)\int_{0}^{\infty}\beta(a)i(t,a)da.
		\end{array}
		\right.  \tag{2.2}  \label{eq:2.2}
	\end{align*}
The initial conditions of this system are given by
	\begin{equation*}
		S(t_{0})=S_{0}\geq0, \tag{2.3}  \label{eq:2.3}
	\end{equation*}
    \begin{equation*}
        R(t_{0})=R_{0}\geq0,\tag{2.4}  \label{eq:2.4}
    \end{equation*}
    and%
	\begin{equation*}
		i(t_{0},a)=i_{0}(a)\in L_{+}^{1}(0,+\infty),  \tag{2.5}  \label{eq:2.5}
	\end{equation*}
	where $L_{+}^{1}(0,\infty)$ is the positive cone of non-negative integral function. In this model, \( S(t) \) represents the number of susceptible individuals at time \( t \), and \( t \rightarrow \tau(t) \) denotes the transmission rate at time $t$(as detailed in \cite{ref2}). The parameter \( \nu > 0 \) is the recovery rate of infected individuals, \( D \geq 0 \) is the mortality rate caused by the disease, and \( \delta\geq0 \) represents the rate of immunity loss. The function \( \beta(a) \) gives the probability of being infectious (capable of transmitting the pathogen) at the age of infection \( a \). Additionally, \( R(t) \) denotes the number of recovered individuals at time \( t \). 
	
	We will continue to adopt the following assumptions for the subsequent discussion. 
	
	\begin{assumption}
		\label{ASS2.1} We assume that 
		
		\begin{itemize}
			\item[$(i)$] \textrm{The transmission rate $t \to\tau(t)$ is a
				bounded continuous map from $[t_{0}, +\infty)$ in $[0, +\infty)$$( \left\|\tau(t)\right\|\leq \tau_{max}, \tau_{max}>0)$; }
			
			\item[$(ii)$] \textrm{The probability to be infectious at the age of
				infection $a \to\beta(a) \in L^{\infty}_{+}(0, +\infty)$ is a non-negative
				and measurable function of $a$ which is bounded by $1$.
			}
				\item[$(iii)$] Under ideal conditions $($especially complete pathogen inactivation post-mortem$)$, the domain of the function $\beta(a)$ is a bounded interval, i.e., there exists $a_+>0$ such that $\beta(a)=0, a\in[a_+,+\infty)$.
		\end{itemize}
	\end{assumption}
	
	From the model \hyperref[eq:2.2]{(\ref*{eq:2.2})}, we obtain
	\begin{equation*}
		I^{^{\prime}}(t)=\underset{(I)}{\underbrace{\tau(t)S(t)\int_{0}^{\infty}%
				\beta(a)i(t,a)da}}-\underset{(II)}{\underbrace{(\nu+D)\int_{0}^{\infty
				}i(t,a)da}},  \tag{2.6}  \label{eq:2.6}
	\end{equation*}
	where $(I)$ is the flow of new infected, and $(II)$ is the flow of
	individuals who die and recover. Let
	\begin{equation*}
		I(t_{0})=\int_{0}^{\infty}i_{0}(a)da=:I_{0}\geq0.  \tag{2.7}  \label{eq:2.7}
	\end{equation*} If the total number of individuals in the cohort is \( n \) at time \( t_0 \), we have 
	\begin{equation*}
		S_{0}+I_{0}+R_{0}=n.  \tag{2.8}  \label{eq:2.8}
	\end{equation*}
	
	\subsection*{(b) Volterra integral equation formulation of the model}%
	
	We use $N(t)$ to represent the flow of new infected
	individuals at time $t$. Then
	\begin{equation*}
		N(t)=\tau(t)S(t)\int_{0}^{\infty}\beta(a)i(t,a)da,  \tag{2.9}
		\label{eq:2.9}
	\end{equation*}
	i.e.,
	\begin{equation*}
		N(t)=\tau(t)S(t)C(t).  
	\end{equation*}
	For each $t\geq t_{0}$, using equation \hyperref[eq:2.9]{(\ref*{eq:2.9})} to substitute model \hyperref[eq:2.2]{(\ref*{eq:2.2})}, we obtain
	\begin{equation*}
		S^{\prime}(t)=-N(t)+\delta R(t),  \tag{2.10}  \label{eq:2.10}
	\end{equation*}
	Substituting \hyperref[eq:2.1]{(\ref*{eq:2.1})} into \hyperref[eq:2.2]{(\ref*{eq:2.2})}, we obtain
	\begin{equation*}
		R^{\prime}(t)=\nu I(t)-\delta R(t).  \tag{2.11}  \label{eq:2.11}
	\end{equation*}
	Similarly, combining \hyperref[eq:2.9]{(\ref*{eq:2.9})} and 
	\hyperref[eq:2.6]{(\ref*{eq:2.6})} yields
	\begin{equation*}
		I^{\prime }(t)=N(t)-(\nu +D)I(t).
	\end{equation*}
	
	 We use the method of characteristic lines to solve the second and fourth equations in \hyperref[eq:2.2]{(\ref*{eq:2.2})}. First, we have %
	\begin{equation}
		\left\{ 
		\begin{aligned}
			&\frac{dt}{da}=1, \\ 
			&\frac{di(t,a)}{dt}=\frac{\partial i(t,a)}{\partial t}+\frac{\partial i(t,a)}{%
				\partial a}\frac{da}{dt}=-(\nu+D)i(t,a).%
		\end{aligned}
		\right.  \tag{2.12}  \label{eq:2.12}
	\end{equation}
	Using the method of variable separation for \hyperref[eq:2.12]{(\ref*{eq:2.12})%
	}, we obtain
	\begin{equation*}
		i(t,a)=\left\{ 
		\begin{aligned}
			&N(t-a)e^{-(\nu+D)a},&&a\leq t-t_{0}, \\ 
			&i_{0}(a-(t-t_{0}))e^{-(\nu+D)(t-t_{0})},&&a\geq t-t_{0}.%
		\end{aligned}
		\right.   \tag{2.13}  \label{eq:2.13}
	\end{equation*}
    
	According to equation \hyperref[eq:2.9]{(\ref*{eq:2.9})}, equation 
	\hyperref[eq:2.13]{(\ref*{eq:2.13})} and Assumption \ref{ASS2.1}, we deduce
	that $t\rightarrow N(t)$ satisfies the following Volterra integral equation
	\begin{equation*}
		N(t)=\tau (t)S(t)\int_{t-t_{0}}^{\infty }\beta (a)e^{-(\nu
			+D)(t-t_{0})}i_{0}(a-(t-t_{0}))da+\tau (t)S(t)\int_{0}^{t-t_{0}}\beta
		(a)e^{-(\nu +D)a}N(t-a)da.  \tag{2.14}  \label{eq:2.14}
	\end{equation*}
	Let
	\begin{equation*}
		\Lambda(t)=\int_{t-t_{0}}^{\infty}\beta(a)e^{-(%
			\nu+D)(t-t_{0})}i_{0}(a-(t-t_{0}))da,
	\end{equation*}
thus, 
	\begin{equation*}
		N(t)=\tau (t)S(t)\left[\Lambda (t)+\int_{0}^{t-t_{0}}\beta (a)e^{-(\nu
			+D)a}N(t-a)da\right]=\tau (t)S(t)\left[\Lambda (t)+\int_{t_{0}}^{t}\beta (t-a)e^{-(\nu
			+D)(t-a)}N(a)da\right].
	\end{equation*}
    Based on the deductions from \hyperref[eq:2.1]{(\ref*{eq:2.1})}, \hyperref[eq:2.10]{(\ref*{eq:2.10})}, and \hyperref[eq:2.11]{(\ref*{eq:2.11})}, we know that \( S(t) \) in \hyperref[eq:2.14]{(\ref*{eq:2.14})} is related to \( N(t) \). To explore the solution $N(t)$ of \hyperref[eq:2.14]{(\ref*{eq:2.14})} further, the expression for $S(t)$ will be derived below. Firstly, we transform the expression for \( I(t) \) in \hyperref[eq:2.1]{(\ref*{eq:2.1})} into
	
	\begin{align*}
		I(t) &
		=\int_{t-t_{0}}^{\infty}e^{-(\nu+D)(t-t_{0})}i_{0}(a-(t-t_{0}))da+%
		\int_{0}^{t-t_{0}}e^{-(\nu+D) a}N(t-a)da   \\
		& =e^{-(\nu+D)(t-t_{0})}I_{0}+\int_{0}^{t-t_{0}}e^{-(\nu+D)a}N(t-a)da,  \tag{2.15}  \label{eq:2.15}
	\end{align*}
	by \hyperref[eq:2.7]{(\ref*{eq:2.7})} and \hyperref[eq:2.13]{(\ref*{eq:2.13})%
	}. Due to equation \hyperref[eq:2.4]{(\ref*{eq:2.4})} and \hyperref[eq:2.11]{(\ref*{eq:2.11})}, using the
	method of variable separation, we obtain 
	
	\begin{equation*}
		R(t)=e^{-\delta(t-t_{0})}R_{0}+\nu\int_{t_{0}}^{t}e^{-\delta
			(t-\tau)}I(\tau)d\tau.  \tag{2.16}  \label{eq:2.16}
	\end{equation*}
	Substituting $I(\tau)$ in \hyperref[eq:2.16]{(\ref*{eq:2.16})} by 
	\hyperref[eq:2.15]{(\ref*{eq:2.15})}, we have 
	
	\begin{equation*}
		R(t)=e^{-\delta(t-t_{0})}R_{0}+\nu\int_{t_{0}}^{t}e^{-\delta
			(t-\tau)}\left[e^{-(\nu+D)(\tau-t_{0})}I_{0}+\int_{0}^{\tau-t_{0}}e^{-(\nu
			+D)a}N(\tau-a)da\right]d\tau.  \tag{2.17}  \label{eq:2.17}
	\end{equation*}
	Based on equations \hyperref[eq:2.10]{(\ref*{eq:2.10})} and \hyperref[eq:2.3]{(\ref*{eq:2.3})}, substituting $R(t)$ by \hyperref[eq:2.17]{(%
		\ref*{eq:2.17})}, we obtain 
	\begin{align*}
    S(t) =& S_{0}-\int_{t_{0}}^{t}N(m)-\delta \biggl\{e^{-\delta (m-t_{0})}R_{0}+\\
        &\nu \int_{t_{0}}^{m}e^{-\delta (m-\tau )} \left[ e^{-(\nu +D)(\tau
			-t_{0})}I_{0}+\int_{0}^{\tau -t_{0}}e^{-(\nu +D)a}N(\tau -a)da\right]d\tau\biggr\}dm ,
		\tag{2.18}  \label{eq:2.18}
\end{align*}
    which means that equation \hyperref[eq:2.14]{(\ref*{eq:2.14})} is a nonlinear Volterra integral equation of the second kind about $N(t)$, which is introduced in \cite{ref59}.\newline
	\bigskip 
	\section{\ Kermack-McKendrick model starting from a\ single and
			multiple cohorts of infected patients}
	
	\label{Section3} The functions \( a \rightarrow i_{0}(a) \) and \( a \rightarrow \beta(a) \) play crucial roles in the parameter identifiability of our Kermack-McKendrick model. In studies \cite{ref23,ref27}, the Dirac mass is utilized as the initial condition for the model. However, this type of initial conditions don't satisfy the regularity for classical solutions of the associated partial differential equation. In \cite{ref27}, the authors rigorously define the concept of measure-valued solutions. By employing duality methods and semigroup theory, they establish the existence and uniqueness of these measure-valued solutions. This section aims to use limit distributions to approach the Dirac mass, leading to more reasonable results.
	\subsection*{\textrm{(a) A single cohort initial distribution for the Volterra integral equation}}
    Define the initial distribution 
    \begin{equation*}
        i_{0}(a)=I_0\delta_0(a)\tag{3.1}  \label{eq:3.1},
    \end{equation*}
where $\delta_0(a)$ is called the Dirac mass centered at age $0$. We use the Dirac mass to represent all the infected individuals who have the same age of infection $a = 0$ at time $t_0$. To understand this concept of Dirac mass centered at 0 in equation \hyperref[eq:3.1]{(\ref*{eq:3.1})}, we first consider an approximation with a class of functions,
\begin{equation*}
        i_{\kappa}(a)=I_0\kappa p(\kappa a)\tag{3.2},  \label{eq:3.2}
\end{equation*}
where $p(a)$ is an arbitrarily given non-negative integrable function satisfying
\begin{equation*}
		    \int_{0}^{+\infty}p(a)da=1.\tag{3.3}  \label{eq:3.3}
\end{equation*}
Recall that
\begin{equation*}
		\lim_{\kappa\rightarrow\infty}\int_{a_1}^{a_2}i_{\kappa}(a)da=\lim_{\kappa\rightarrow\infty}\int_{\kappa a_1}^{\kappa a_2}I_0p(s)ds=\left\{ 
		\begin{array}{l}
			0, \,\,\,\,\text{if}\,\,\,\,a_2>a_1>0,\\ 
			I_0,\,\,\text{if}\,\,\,\,a_2>a_1=0.
		\end{array}
		\right. 
	\end{equation*}
In other words, when $\kappa$ tends to $+\infty$,  the initial distribution of population $i_{\kappa}(a)$ is approaching
the case where all the infected individuals at time $t_0$ have the same age of infection $a = 0$.

To derive model \hyperref[eq:2.2]{(\ref*{eq:2.2})} with Dirac mass initial distribution as limit, we look back at the expression of $\Lambda(t)$ in equation \hyperref[eq:2.14]{(\ref*{eq:2.14})}
\begin{equation*}
	\Lambda(t)=e^{-(\nu+D)(t-t_{0})}\int_{0}^{+\infty}%
	\beta(a+(t-t_{0}))i_{0}(a)da.  \tag{3.4}  \label{eq:3.4}
\end{equation*}
Substituting $i_0(a)$ in equation \hyperref[eq:3.4]{(\ref*{eq:3.4})} with $%
i_\kappa(a)$ in equation \hyperref[eq:3.2]{(\ref*{eq:3.2})}, we obtain 
\begin{equation*}
	\Lambda_{\kappa}(t)=e^{-(\nu+D)(t-t_{0})}\int_{0}^{+\infty}\beta
	(a+(t-t_{0}))I_{0}\kappa p(\kappa a)da.  \tag{3.5}  \label{eq:3.5}
\end{equation*}

\begin{lemma}
\label{lem4.2}
	\textrm{Let Assumption \ref{ASS2.1} be satisfied, and assume in addition
		that $a\rightarrow\beta(a)$ is continuous. Then we have }
\begin{equation*}
	\underset{\kappa\rightarrow\infty}{\lim}\Lambda_{\kappa}(t)=I_{0}e^{-(%
		\nu+D)(t-t_{0})}\beta(t-t_{0}),  \tag{3.6}  \label{eq:3.6}
\end{equation*}
where the limit is uniform in $t\geq t_{0}$. That is 
\begin{equation*}
	\underset{\kappa\rightarrow\infty}{\lim}\underset{t\geq t_{0}}{\sup}%
	\left|\Lambda_{\kappa}(t)-I_{0}e^{-(\nu+D)(t-t_{0})}\beta(t-t_{0})\right|=0.
\end{equation*}
\end{lemma}
The proof of Lemma \ref{lem4.2} is in Appendix \ref{app:A}. Define
\begin{equation}
	\Gamma(a)=e^{-(\nu+D)a}\beta(a),\forall a\geq0.  \tag{3.7}  \label{eq:3.7}
\end{equation}
If the function $\beta(a)$ is continuous and the initial distribution is \hyperref[eq:3.1]{(\ref*{eq:3.1})}, based on \cite{ref58}, equation \hyperref[eq:2.14]{(\ref*{eq:2.14})} becomes
\begin{equation*}
	N(t) =\tau(t) S(t) \left[ I_0 \times \Gamma\left(t-t_0\right) +
	\int_0^{t-t_0} \Gamma(a) \, N(t-a) da \right], \forall t \geq t_0.
	\tag{3.8}  \label{eq:3.8}
\end{equation*}
Recall that the domain of $i_0(a)$ is $L_{+}^{1}(0,+\infty)$ in equation \hyperref[eq:2.5]{(\ref*{eq:2.5})}, then we will extend the initial condition to a Dirac mass with the help of $i_{\kappa}(a)$. Applying equation \hyperref[eq:2.14]{(\ref*{eq:2.14})} and \hyperref[eq:2.18]{(\ref*{eq:2.18})}, the Kermack-McKendrick model can be
reformulated for $t\geq t_{0}$ as 
\begin{equation*}
	N_{\kappa }(t)=\tau (t)S_{\kappa }(t)\left[\Lambda _{\kappa
	}(t)+\int_{0}^{t-t_{0}}\beta (a)e^{-(\nu +D)a}N_{\kappa }(t-a)da\right] , \tag{3.9}
	\label{eq:3.9}
\end{equation*}%
where $\Lambda _{\kappa }$ is defined in \hyperref[eq:3.5]{(\ref*{eq:3.5})}
and 

\begin{align*}
    S_{\kappa}(t) =& S_{0} - \int_{t_{0}}^{t} N_{\kappa}(m) - \delta \biggl\{ e^{-\delta (m-t_{0})} R_{0} + \\
        &\nu \int_{t_{0}}^{m} e^{-\delta (m-\tau)} \left[ e^{-(\nu + D)(\tau - t_{0})} I_{0} + \int_{0}^{\tau - t_{0}} e^{-(\nu + D)a} N_{\kappa}(\tau - a) \, da \right] d\tau \biggr\} dm. \tag{3.10} \label{eq:3.10}
\end{align*}

\begin{theorem}
\label{Th47}
	\textrm{If Assumption \ref{ASS2.1} and the conditions in Appendix \ref{app:B} and Lemma \ref{lem4.2} are satisfied, then we have }
	
	\textrm{%
		\begin{equation}
			N(t)=\underset{\kappa\rightarrow\infty}{\lim}N_{\kappa}(t),  \tag{3.11}
			\label{eq:3.11}
		\end{equation}
		where the limit is uniform in $t$ on every closed and bounded interval of $%
		[t_{0},+\infty)$, and the map $t\rightarrow N(t)$ is the unique continuous
		solution of the Volterra integral equation \hyperref[eq:3.8]{(\ref*{eq:3.8})}. }
\end{theorem}

	\begin{proof}
    To improve the readability of the subsequent content, we introduce a
new bounded continuous operator $F$, such that the operator $F$ acting on $N(t)$ produces $S(t)$, i.e., 
\begin{align*}
    F(N(t)) =& S_{0}-\int_{t_{0}}^{t}N(m)-\delta \biggl\{e^{-\delta (m-t_{0})}%
	R_{0}+\\
        &\nu \int_{t_{0}}^{m}e^{-\delta (m-\tau )}\left[e^{-(\nu +D)(\tau
		-t_{0})}I_{0}+\int_{0}^{\tau -t_{0}}e^{-(\nu +D)a}N(\tau -a)da\right]d\tau \biggr\}dm,
	\tag{3.12}  \label{eq:3.12}
\end{align*}
namely,
\begin{equation*}
	S(t)=F(N(t)),\,\,S_{\kappa}(t)=F(N_{\kappa}(t)),
\end{equation*}
and we rewrite equation \hyperref[eq:3.8]{(\ref*{eq:3.8})} as equation
\begin{equation}
	x(t)=\tau (t)F(x(t))\left[I_{0}\times\Gamma (t-t_{0})+\int_{0}^{t-t_{0}}\Gamma
	(a)x(t-a)da\right], \tag{3.13}  \label{eq:3.13}
\end{equation}
where $x(t)$ denotes the solution to equation \hyperref[eq:3.13]{(\ref*{eq:3.13})}.
		Define
		\begin{equation}
			X(x)(t)=\tau(t)F(x(t))\left[I_{0}\Gamma(t-t_{0})+\int_{0}^{t-t_{0}}\Gamma
			(a)x(t-a)da\right].  \tag{3.14}  \label{eq:3.14}
		\end{equation}
        Theorem \ref{Th47} only needs to be proven to hold on a closed and bounded interval. $\exists\, T>0$ such that $\Omega \subset C([t_{0}, t_{0}+T])$ be a bounded closed set and
         \begin{equation*}
             \forall x\in C([t_{0}, t_{0}+T]), ||x||\leq 2\tau_{max}S_0I_0.
         \end{equation*}
         \textrm{Based on Appendix \ref{app:B},
	we observe that both $N(t)$ in \hyperref[eq:3.8]{(\ref*{eq:3.8})} and $A(t)
	$ in Appendix \ref{app:C}  satisfy equation \hyperref[eq:3.13]{(\ref*{eq:3.13})}. In the following part, we will show the conditions where the solution to equation \hyperref[eq:3.13]{(\ref*{eq:3.13})} is unique and model \hyperref[eq:2.2]{(\ref*{eq:2.2})} with a single cohort of infected extending the earlier model of Kermack-Mckendrick with initial distribution $L_{+}^{1}(0,+\infty)$}.
         The existence and uniqueness of the solution to the integral equation \hyperref[eq:3.14]{(\ref*{eq:3.14})} is equivalent to prove that the operator \( X \) has a unique fixed point in the space \( \Omega \), where \( T > 0 \). 
		
		\textbf{Step1:} To prove that the operator $X:\Omega\rightarrow\Omega$. Based on Assumption \ref{ASS2.1} and the conditions in Appendix \ref{app:B} and Appendix \ref{app:C}, the mappings \( \tau(t) \), \( \Gamma(t) \), \( N(t) \), and \( F(\cdot) \) are continuous and bounded. According to equation \hyperref[eq:3.14]{(\ref*{eq:3.14})}, \( X(x)(t) \) is continuous in the interval \([t_{0}, t_{0}+T]\) for \( T > 0 \). Based on equation \hyperref[eq:3.14]{(\ref*{eq:3.14})},
        \begin{equation*}
            \left\Vert X(x) \right\Vert\leq \tau_{max}S_0I_0+2\tau_{max}S_0I_0T,
        \end{equation*}
        hence, if $T\leq\frac{1}{2}$, \( X: \Omega \rightarrow \Omega \) is proven.\\
		\textbf{Step2: }To prove that $X$ is a contraction mapping. Based on equation \hyperref[eq:3.12]{(\ref*{eq:3.12})}, we get
		\begin{align*}
			\left\Vert F(x_{1})-F(x_{2})\right\Vert  
            & =\Vert-\int_{t_{0}}^{t}%
			x_{1}(m)-\delta\left\{\nu\int_{t_{0}}^{m}e^{-\delta(m-\tau)}\left[\int_{0}^{\tau-t_{0}%
			}e^{-(\nu+D)a}x_{1}(\tau-a)da\right]d\tau\right\}dm\\
			&  +\int_{t_{0}}^{t}x_{2}(m)-\delta\left\{\nu\int_{t_{0}}^{m}e^{-\delta(m-\tau
				)}\left[\int_{0}^{\tau-t_{0}}e^{-(\nu+D)a}x_{2}(\tau-a)da\right]d\tau\right\}dm\Vert\\
			&  \leq(1+\nu\delta T^{2})T\left\Vert x_{1}-x_{2}\right\Vert.
		\end{align*}
         Therefore, $F(x)$ is Lipschitz continuous with respect to $x$. Let
		\[
		L_{F}:=1+\nu\delta T^{2}>0.
		\]
		Then we obtain
		\[
		\left\Vert F(x_{1})-F(x_{2})\right\Vert \leq L_{F}T\left\Vert x_{1}%
		-x_{2}\right\Vert.
		\]
		 For all $x_{1}, x_{2}\in\Omega$, we have
		
		\begin{align*}
			\left\Vert X(x_{1})-X(x_{2})\right\Vert &\leq L_{F}T\tau_{max} (I_{0}+T
\left\Vert x_{1}\right\Vert )\left\Vert
			x_{1}-x_{2}\right\Vert +T\tau_{max} \left\Vert F(x_{2}%
			)\right\Vert \left\Vert x_{1}-x_{2}\right\Vert\\
            &=\left[L_{F}(I_{0}+T\left\Vert x_{1}\right\Vert )+\left\Vert F(x_{2}%
			)\right\Vert \right]T \tau_{max} \left\Vert x_{1}-x_{2}\right\Vert\\
            &\leq \left[L_{F}(I_{0}+2T\tau_{max}S_0I_0)+S_0 \right]T \tau_{max} \left\Vert x_{1}-x_{2}\right\Vert
		\end{align*}
        There exists $\alpha>0$ such that whenever $T<\min\left\{\frac{1}{2},\alpha\right\}$ holds, 
		
		\[
		\left[L_{F}(I_{0}+2T\tau_{max}S_0I_0)+S_0 \right]T\tau_{max}<1.
		\]
		Thus, $X$ is a contraction mapping. Finally, according to the Banach contraction mapping theorem, there is a unique solution $x$ in equation \hyperref[eq:3.13]{(\ref*{eq:3.13})}, i.e.,
		\[
		X(x)=x,
		\]
		Therefore, \hyperref[eq:3.11]{(\ref*{eq:3.11})} is proved.
	\end{proof}
		Suppose the Kermack-McKendrick model starts from a single cohort of infected individuals. This means the initial distribution comprises $I_0$ individuals, all with an infection age $a=0$ at time $t_0$. In this case, the flow of new infections $t \rightarrow N(t)$ is the unique continuous solution of the following Volterra integral equation
		\begin{align*}
			N(t) = \tau(t) S(t) \left[I_{0} \Gamma(t - t_{0}) + \int_{0}^{t - t_{0}} \Gamma(a) N(t - a) \, da \right],\tag{3.15} \label{eq:3.15}
		\end{align*}
		where
        \begin{align*}
    S(t) =& S_{0}-\int_{t_{0}}^{t}N(m)-\delta\biggl\{e^{-\delta(m-t_{0})}%
			R_{0}+\\
        &\nu\int_{t_{0}}^{m}e^{-\delta(m-\tau)}\left[e^{-(\nu+D)(\tau
				-t_{0})}I_{0}+\int_{0}^{\tau-t_{0}}e^{-(\nu+D)a}N(\tau-a)da\right]d\tau\biggr\}dm.\tag{3.16}  \label{eq:3.16}
\end{align*}
	
	\begin{remark}
		When the initial distribution is a Dirac mass centered at $a=0$, the
		total number of infectious individuals at time $t$ is 
		\begin{equation*}
			C(t)=I_{0}\beta(t-t_{0})e^{-(\nu+D)(t-t_{0})}+\int_{0}^{t-t_{0}}%
			\beta(a)e^{-(\nu+D)a}N(t-a)da,\forall t\geq t_{0},
		\end{equation*}
		and the total number of infected individuals at time $t$ is 
		\begin{equation*}
			I(t)=I_{0}e^{-(\nu+D)(t-t_{0})}+\int_{0}^{t-t_{0}}e^{-(\nu+D)a}N(t-a)da,%
			\forall t\geq t_{0}.
		\end{equation*}
    	\end{remark}
		\subsection*{\textrm{(b) Multiple cohorts initial distribution for the model}}
		In the case of multiple cohorts, the initial distribution becomes 
		\begin{equation*}
			i_{0}(a)=I_{0}^{1}\delta_{a_{1}}(a)+I_{0}^{2}\delta_{a_{2}}(a)+...+I_{0}^{n}%
			\delta_{a_{n}}(a),\tag{3.17}  \label{eq:3.17}
		\end{equation*}
		where $0\leq a_{1}<a_{2}<...<a_{n}$ are the ages of infection for each cohort at
		time $t_{0}$, and $I_{0}^{j}$ is the number of infected individuals in the $j^{th}$-cohort at time $t_{0}$. Then the corresponding $i_{\kappa}(a)$ is
		\begin{equation*}
			i_{\kappa}(a)=\sum_{j=1}^{n}I_{0}^{j}\kappa p(\kappa (a-a_j))\mathds{1}_{(a-a_{j})}, 
		\end{equation*}
		where $p(a)$ is a non-negative function satisfying
		\begin{equation*}
		    \int_{0}^{+\infty}p(a)da=1,
		\end{equation*}
        and
		\begin{equation*}
			\mathds{1}_{(a-a_{j})}=\left\{
			\begin{array}{c}
				1, \,a \geq a_{j}, \\ 
				0, \,a < a_{j}.
			\end{array}
            \right.
		\end{equation*}
	Similar to the case of a single cohort in \hyperref[eq:3.5]{(\ref*{eq:3.5})}, we define
		
		\begin{equation*}
			\Lambda_{\kappa}^{j}(t):=
			I_{0}^{j}e^{-(\nu+D)(t-t_{0})}\int_{0}^{\infty}\beta(a+(t-t_{0}))\kappa
			p(\kappa(a-a_j))\mathds{1}_{(a-a_{j})}da.
		\end{equation*}
	
	\begin{theorem}
    \label{TH4.8}
		\textrm{If $\beta(a)$ is continuous and  \hyperref[eq:3.3]{(\ref*{eq:3.3})} is satisfied, then}
	\end{theorem}
	
	\begin{equation*}
		\underset{\kappa\rightarrow+\infty}{\lim}\Lambda_{%
			\kappa}^{j}(t)=I_{0}^{j}e^{-(\nu+D)(t-t_{0})}\beta(a_{j}+(t-t_{0})).
	\end{equation*}
	
	\begin{proof}
		The case $I_{0}^{j}=0$ is obviously valid, and below we only discuss the case $I_{0}^{j}>0$. We firstly obtain
		\begin{align*}
			&  I_{0}^{j}e^{-(\nu+D)(t-t_{0})}\int_{0}^{+\infty}\beta(a+(t-t_{0}))\kappa
			p(\kappa (a-a_j))\mathds{1}_{(a-a_{j})}da-I_{0}^{j}e^{-(\nu+D)(t-t_{0})}\beta(a_{j}%
			+(t-t_{0}))\\
			 =&I_{0}^{j}e^{-(\nu+D)(t-t_{0})}\int_{-a_{j}}^{+\infty
			}\beta(a_{j}+s+(t-t_{0}))\kappa p(\kappa s)\mathds{1}_{(s)}ds-I_{0}^{j}e^{-(\nu+D)(t-t_{0})}\int_{0}^{+\infty}\beta(a_{j}+(t-t_{0}%
			))\kappa p(\kappa s)ds\\
			  =&I_{0}^{j}e^{-(\nu+D)(t-t_{0})}\int_{0}^{+\infty}\beta(a_{j}+s+(t-t_{0}%
			))\kappa p(\kappa s)ds-I_{0}^{j}e^{-(\nu+D)(t-t_{0})}\int_{0}^{+\infty}%
			\beta(a_{j}+(t-t_{0}))\kappa p(\kappa s)ds\\
			  =&I_{0}^{j}e^{-(\nu+D)(t-t_{0})}\int_{0}^{\eta}\left[\beta(a_{j}+s+(t-t_{0}%
			))-\beta(a_{j}+(t-t_{0}))\right]\kappa p(\kappa s)ds+\\
			& I_{0}^{j}e^{-(\nu+D)(t-t_{0})}\int_{\eta}^{+\infty}\left[\beta(a_{j}%
			+s+(t-t_{0}))-\beta(a_{j}+(t-t_{0}))\right]\kappa p(\kappa s)ds.
		\end{align*}
		For all $\varepsilon>0$, there exists $t_{1}>t_{0}$ such that 
		
		\[
		I_{0}^{j}e^{-(\nu+D)(t-t_{0})}\underset{a\geq0}{\sup}\beta(a)\leq
		\frac{\varepsilon}{4},\,\forall t\in[t_{1},+\infty).
		\]
		Let $\eta>0$ be such that %
		
		\[
		a\leq\eta\Rightarrow|\beta(a+t)-\beta(t)|\leq\frac{\varepsilon}{2I_{0}^{j}},\forall
		t\in\lbrack t_{0},t_{1}].
		\]
		Then we have %
        \begin{equation*}
    \begin{aligned}
        &\left|I_{0}^{j}e^{-(\nu+D)(t-t_{0})}\int_{0}^{\eta}[\beta(a_{j}+s+(t-t_{0}))-\beta(a_{j}+(t-t_{0}))]\kappa p(\kappa s)ds\right| \\
        &\leq \left\{  
            \begin{aligned}
                &I^{j}_{0}e^{-(\nu+D)(t-t_{0})}2\sup_{a\geq0}\beta(a),  && \text{if } t \in [t_{1},+\infty), \\
                &I^{j}_{0}e^{-(\nu+D)(t-t_{0})}\int_{0}^{\eta}\frac{\varepsilon}{2I_{0}^{j}}\kappa p(\kappa a)da, && \text{if } t \in [t_{0},t_{1}].
            \end{aligned}
        \right. 
    \end{aligned}
\end{equation*}
		Therefore,%
		\begin{align*}
    &\left|\Lambda_{\kappa}^{j}(t)-I_{0}^{j}e^{-(\nu+D)(t-t_{0})}\beta(a_{j}%
			+(t-t_{0}))\right| \\
        &\leq\frac{\varepsilon}{2}+\left|I_{0}^{j}e^{-(\nu+D)(t-t_{0})}\int_{\eta}^{\infty
			}\left[\beta(a_{j}+s+(t-t_{0}))-\beta(a_{j}+(t-t_{0}))\right]\kappa p(\kappa s)ds\right|.
\end{align*}
		According to
        \begin{align*}
    &\left|I_{0}^{j}e^{-(\nu+D)(t-t_{0})}\int_{\eta}^{\infty}\left[\beta(a_{j}%
			+s+(t-t_{0}))-\beta(a_{j}+(t-t_{0}))\right]\kappa p(\kappa s)ds\right| \\
        &\leq2I^{j}_{0}e^{-(\nu+D)(t-t_{0})}\underset{a\geq0}{\sup}\beta(a)\left(1-\int_{0
			}^{\eta}\kappa p(\kappa a)da\right).
\end{align*}
		We obtain
		\begin{align*}
			&  2I^{j}_{0}e^{-(\nu+D)(t-t_{0})}\underset{a\geq0}{\sup}\beta(a)\left(1-\int_{0
			}^{\eta}\kappa p(\kappa a)da\right)=2I^{j}_{0}e^{-(\nu+D)(t-t_{0})}\underset{a\geq0}{\sup}\beta(a)\left(1-\int_{0
			}^{\kappa\eta}p(a)da\right)\rightarrow0,as\text{ }\kappa\rightarrow\infty.
		\end{align*}
	    Equivalently,
		\[
		\Lambda_{\kappa}^{j}(t)\rightarrow I_{0}%
		^{j}e^{-(\nu+D)(t-t_{0})}\beta(a_{j}+(t-t_{0}%
		))=\Gamma(a_{j}+(t-t_{0}))\frac{I_{0}^{j}%
		}{e^{-(\nu+D)a_{j}}},\kappa\rightarrow+\infty.
		\]
		
	\end{proof}
	
		 Assume that the initial distribution of infected individuals consists of $n>1$ cohorts
		of infected individuals with age of infection $a_{1}<a_{2}<...<a_{n}$ at time $t_{0}$, that is,
		\begin{equation*}
			i_{0}(a)=I_{0}^{1}\delta_{a_{1}}(a)+I_{0}^{2}\delta_{a_{2}}(a)+...+I_{0}^{n}%
			\delta_{a_{n}}(a),
		\end{equation*}
		where $I_{0}^{j}>0$ is the number of infected in the $j^{th}$-cohort at time $%
		t_{0}$. Then the flow of infected $t\rightarrow N(t)$ satisfies the
		following Volterra integral equation,
		\begin{equation*}
			N(t)=\tau(t)S(t)\left[\underset{j=1}{\overset{n}{\sum}}\Gamma(a_{j}+(t-t_{0}))%
			\frac{I_{0}^{j}}{e^{-(\nu+D)a_{j}}}+\int_{0}^{t-t_{0}}\Gamma(a)N(t-a)da\right], \forall t\geq t_0,\tag{3.18}  \label{eq:3.18}
		\end{equation*}
		where 
        \begin{align*}
    S(t) =& S_{0}-\int_{t_{0}}^{t}N(m)-\delta\biggl\{e^{-\delta (m-t_{0})}R_{0}+\\
        &\nu \int_{t_{0}}^{m}e^{-\delta (m-\tau )}\left[e^{-(\nu +D)(\tau
				-t_{0})}\underset{j=1}{\overset{n}{\sum}}%
			I_{0}^{j}+\int_{0}^{\tau -t_{0}}e^{-(\nu +D)a}N(\tau -a)da\right]d\tau\biggr\}dm .\tag{3.19}  \label{eq:3.19}
\end{align*}
	
	\subsection*{\textrm{(c) Effective reproduction number}}
    The basic reproduction number \( \mathcal{R}_0 \) is a key parameter in infectious disease models. Based on \cite{ref60}, the basic reproduction number \( \mathcal{R}_0 \) is defined as the average number of new cases of an infection caused by one typical infected individual, in a population consisting of susceptibles only. As studied in \cite{ref2}, $\mathcal{R}_0$ can be derived from the daily reproduction number \( \mathcal{R}_0(a) \) identified during the early stages of infectious disease transmission. 
    
    Another important parameter is the effective reproduction number $\mathcal{R}(t)$ (i.e., the instantaneous reproduction number at calendar time $t$), which represents the average number of people someone infected at time $t$ could expect to infect  if conditions remained unchanged (for \( t \geq t_0 \)) based on \cite{ref51}. As the disease progresses, this quantity reflects the impact of herd immunity and intervention measures, such as quarantine. To identify the effective reproduction number $\mathcal{R}(t)$, we study the reproductive power $\mathcal{R}(t,a)$ which is the rate at which an infectious individual at calendar time $t$ and infection-age $a$ produces secondary cases. Based on \cite{ref25}, we decompose the reproductive power \( \mathcal{R}(t,a) \) into the following form,
	\begin{equation*}
		\mathcal{R}(t,a)= \underset{ \mathrm{(A)}}{\underbrace{ \tau(t) }} \times\underset{ 
			\mathrm{(B)}}{\underbrace{ S(t) }} \times\underset{ \mathrm{(C)}}{%
			\underbrace{ \beta(a) }} \times\underset{ \mathrm{(D)}}{\underbrace{ e^{-(\nu+D)
					a} }},\,\forall
		a\in[0,a_+],  \tag{3.20}  \label{eq:3.20}
 	\end{equation*}
	where (A) is the rate of transmission at time $t$, (B) is the number of susceptible individuals at time $t$, (C) is the probability that infected individuals of infection age $a$ who are infectious (capable of transmitting the virus),
	(D) is the probability for an infected individual for $a$ days to still be infected. According to equation \hyperref[eq:3.7]{(\ref*{eq:3.7})}, equation \hyperref[eq:3.20]{(\ref*{eq:3.20})} can be rewritten as 	
	\begin{equation*}
		\mathcal{R}(t,a)=\tau(t)S(t)\Gamma(a),\,\forall
		a\in[0,a_+],  \tag{3.21}  \label{eq:3.21}
	\end{equation*}
	where $a$ denotes the age of infection. \\
	{\textbf{Effective reproduction number:}}
		The effective reproduction can be expressed as the integral of $\mathcal{R}(t,a)$ over all ages of infection:
		
		\begin{equation*}
			\mathcal{R}(t)=\int_{0}^{\infty}\mathcal{R}(t,a)da.\tag{3.22}  \label{eq:3.22}
		\end{equation*}
		If model \hyperref[eq:2.2]{(\ref*{eq:2.2})} starts from a single cohort of infected individuals composed of $I_0$ with an age of infection $a=0$ at time $t_0$ in \hyperref[eq:3.15]{(\ref*{eq:3.15})}, based on \hyperref[eq:3.21]{(\ref*{eq:3.21})}, we obtain 
		\begin{equation*}
			N(t)=\mathcal{R}(t,t-t_{0})I_{0}+\int_{0}^{t-t_0} \mathcal{R}(t,s)N(t-s)ds,\forall t\geq
			t_{0}.  \tag{3.23}  \label{eq:3.23}
		\end{equation*}
        
    \section{Connecting the data with the model}
    \label{Section4} 

	\subsection*{(a) The cumulative number of deaths $CD(t)$}
    Real world disease data is filled with uncertainties stemming from asymptomatic transmission, testing inaccuracies, and inconsistent reporting, as noted in \cite{ref5,ref29}. In contrast, mortality data is less affected by distortions caused by variations in diagnostic criteria, case definitions, and human factors, making it more unambiguous. We define a function \( t \rightarrow CD(t) \) as the cumulative number of deaths at time $t$, which is at least twice continuously differentiable. Following the study of \cite{ref55}, we define the cumulative number of deaths using 
	\begin{equation}
		CD^{\prime}(t)=D\int_{0}^{\infty}i(t,a)da.  \tag{4.1}  \label{eq:4.1}
	\end{equation}
    Substituting \hyperref[eq:2.15]{(\ref*%
		{eq:2.15})} into \hyperref[eq:4.1]{(\ref*{eq:4.1})}, we obtain
	\begin{equation*}
		CD^{\prime}(t)=D I(t)=D\left[
		e^{-(\nu+D)(t-t_{0})}I_{0}+\int_{t_{0}}^{t}e^{-(\nu+D)(t-\tau)}N(\tau)d\tau\right],
		\tag{4.2}  \label{eq:4.2}
	\end{equation*}
	and
	\begin{equation}
		I(t)=\frac{CD^{\prime}(t)}{D}.  \tag{4.3}  \label{eq:4.3}
	\end{equation}
	\textrm{By choosing $t=t_{0}$ we obtain }
	
	\begin{equation}
		I_{0}=\frac{CD^{\prime}(t_{0})}{D}  .\tag{4.4}  \label{eq:4.4}
	\end{equation}
From \hyperref[eq:4.2]{(\ref*{eq:4.2})}, we
	have 
	\begin{equation*}
		\int_{t_{0}}^{t}e^{(\nu+D)\sigma}N(\sigma)d\sigma=\frac{e^{(\nu+D)t}CD^{%
				\prime}(t)}{D}-e^{(\nu+D)t_{0}}I_{0}.\tag{4.5}  \label{eq:4.5}
	\end{equation*}
    Differentiating both sides of \hyperref[eq:4.5]{(\ref*{eq:4.5})}, we obtain
	\begin{equation}
		N(t)=\frac{(\nu+D)CD^{\prime}(t)+CD^{\prime\prime}(t)}{D},\forall t\geq
		t_{0}  \tag{4.6}  \label{eq:4.6}.
	\end{equation}
    By fitting the function $CD'(t)$, we can obtain $CD''(t)$ and then derive the explicit expression for $N(t)$. According to equations \hyperref[eq:2.8]{(\ref*{eq:2.8})} and \hyperref[eq:4.4]{(\ref*{eq:4.4})}, the expression of $R_{0}$ is 
	
	\begin{equation}
		R_{0}=n-S_{0}-I_{0}=n-S_{0}-\frac{CD^{\prime}(t_{0})}{D} 
		\tag{4.7}  \label{eq:4.7}.
	\end{equation}
	\textrm{Combining equation \hyperref[eq:4.7]{(\ref*{eq:4.7})} and \hyperref[eq:2.11]{(%
			\ref*{eq:2.11})} and using the method of constant variation, the
		expression for $R(t)$ is }
	\begin{equation*}
		R(t)=e^{-\delta(t-t_{0})}\left(n-S_{0}-\frac{CD^{\prime}(t_{0})}{D}%
		\right)+\int_{t_{0}}^{t}\frac{\nu e^{-\delta(t-\tau)}CD^{\prime}(\tau)}{D}d\tau 
		\tag{4.8}  \label{eq:4.8}.
	\end{equation*}
	\textrm{Connecting \hyperref[eq:2.10]{(\ref*{eq:2.10})} with \hyperref[eq:4.8]%
		{(\ref*{eq:4.8})}, we get }
	
	\begin{equation*}
		S^{\prime}(t)=-\frac{(\nu+D)CD^{\prime}(t)+CD^{\prime\prime}(t)}{D}%
		+\delta e^{-\delta(t-t_{0})}\left(n-S_{0}-\frac{CD^{\prime}(t_{0})}{D}\right)+\delta\int_{t_{0}}^{t}\frac {\nu e^{-\delta(t-\tau)}CD^{\prime}(\tau)}{D}d\tau.
	\end{equation*}
	\textrm{Integrating both sides of the above equation from $t_{0}$ to $t$, the
		number of susceptible individuals is }
	
	\begin{equation*}
		S(t)=S_{0}-\int_{t_{0}}^{t}\left\{\frac{(\nu+D)CD^{\prime}(\sigma)+CD^{\prime
				\prime}(\sigma)}{D}-\delta e^{-\delta(\sigma-t_{0})}(n-S_{0}-\frac{%
			CD^{\prime}(t_{0})}{D})-\delta\int_{t_{0}}^{\sigma}\frac{%
			\nu e^{-\delta(\sigma-\tau)}CD^{\prime}(\tau )}{D}d\tau\right\}d\sigma
            \tag{4.9}  \label{eq:4.9}.
	\end{equation*}
	
	\subsection*{(b) Key parameters $\beta(a)$ and $\tau(t)$}
   The parameter $\beta(a)$ is associated with both the host and the pathogen. In \cite{ref21,ref50}, the unimodal distribution is used as an example. This pattern reflects that a pathogen initially proliferates within the host, followed by a decline due to immune suppression or host death. Current biological researches focus on the dynamics of pathogen load, specifically on pathogen shedding. Due to variations in host physiological conditions, viral load (i.e., viral RNA levels) is used to represent \(\beta(a)\) based on \cite{ref33}. However, more complex distribution patterns may exist; for example, a bimodal distribution \cite{ref48} is associated with early- and late-stage HIV transmission, while a multimodal distribution \cite{ref49} is linked to malaria.

  In summary, viral load exhibits a variety of temporal patterns based on \cite{ref2,ref50,ref51}. Although the function \(\beta(a)\) may vary in different biological studies, it still satisfies the condition $(iii)$ in Assumption \ref{ASS2.1}. 

Next, the expression for $\tau(t)$ is divided into two cases.
When model \hyperref[eq:2.2]{(\ref*{eq:2.2})} starts with a single cohort of infected individuals, the expression for $\tau(t)$ is given by
\begin{equation*}
    \tau(t)=\frac{N(t)}{S(t)\left[ I_0 \Gamma(t-t_0) +
	\int_0^{t-t_0} \Gamma(a) \, N(t-a) da \right]},
\end{equation*}
which is derived from \hyperref[eq:3.15]{(\ref*{eq:3.15})}.
Combining cumulative death data $CD(t)$ with model \hyperref[eq:2.2]{(\ref*{eq:2.2})} and utilizing the expressions for $I_0$ and $N(t)$, provided in \hyperref[eq:4.4]{(\ref*{eq:4.4})} and \hyperref[eq:4.6]{(\ref*{eq:4.6})}, the expression for $\tau(t)$ is given by
\begin{equation*}
    \tau(t)=\frac{(\nu+D)CD'(t)+CD''(t)}{D S(t)\left[ \frac{CD'(t_{0})}{D} e^{-(\nu+D)(t-t_0)}\beta(t-t_0) +
	\int_0^{t-t_0} e^{-(\nu+D)a}\beta(a)\frac{(\nu+D)CD'(t-a)+CD''(t-a)}{D} da \right]},
\end{equation*}
where the expression of $S(t)$ is in equation \hyperref[eq:4.9]{(\ref*{eq:4.9})}.
When the model starts with multiple cohorts of infected individuals, given the values of $a_j$ and $I_0^j$ in equation \hyperref[eq:3.19]{(\ref*{eq:3.19})}, the expression for $\tau(t)$ can be 
\begin{equation*}
    \tau(t)=\frac{(\nu+D)CD'(t)+CD''(t)}{D S(t)[\underset{j=1}{\overset{n}{\sum}}\Gamma(a_{j}+(t-t_{0}))\frac{I_{0}^{j}}{e^{-(\nu+D)a_{j}}}+\int_{0}^{t-t_{0}}\Gamma(a)\frac{(\nu+D)CD'(t-a)+CD''(t-a)}{D} da]}, \tag{4.10}  \label{eq:4.10}
\end{equation*}
which is derived from the expression of $N(t)$ in \hyperref[eq:4.6]{(\ref*{eq:4.6})} and the expression of $S(t)$ in \hyperref[eq:4.9]{(\ref*{eq:4.9})}.
	\section{\textrm{Day by day Kermack–McKendrick model \hyperref[eq:2.2]{(\ref*{eq:2.2})}}}
	
	\label{Section5} 
    This section focuses on discretizing the Kermack-McKendrick model by taking discrete time steps \( \triangle t = 1 \), in other words, \( t = t_0 + k \), where \( k \in \mathbb{N} \). We derive discrete expressions for the functions \( N(t) \), \( R(t) \), \( S(t) \), \( \mathcal{R}(t) \) and \( \mathcal{R}(t, a) \) (where \( a = 0, 1, 2, \ldots, \lfloor a_+ \rfloor \)). This approach can also be extended to accommodate other time intervals. According to equation \hyperref[eq:2.10]{(\ref{eq:2.10})}, the temporal evolution of the susceptible population \( S(t) \) is given by
	\begin{equation*}
		S(t)=S_{0}-\underset{d=t_{0}}{\overset{t-1}{\sum}}%
		\left[N(d)-\delta R(d)\right].
	\end{equation*}
	According to equation \hyperref[eq:2.14]{(\ref%
		{eq:2.14})}, the flow of new infected
	individuals at time $t$ follows the
	discrete time Volterra integral equation 
	\begin{align*}
		N(t) & =\tau(t)S(t)\left[\underset{d=t-t_{0}}{\overset{\infty}{\sum}}%
		\beta(d)e^{-(\nu+D)(t-t_{0})}i_{0}(d-(t-t_{0}))+\underset{d=1}{%
			\overset{t-t_0}{\sum}}\beta(d)e^{-(\nu+D) d}N(t-d)\right] \\
		& =\tau(t)S(t)\left[\underset{d=t-t_{0}}{\overset{\infty}{\sum}}\Gamma (d)\frac{%
			i_{0}(d-(t-t_{0}))}{e^{-(\nu+D)(d-(t-t_{0}))}}+\underset{d=1}{%
			\overset{t-t_0}{\sum}}\Gamma(d)N(t-d)\right].  \tag{5.1}
		\label{eq:5.1}
	\end{align*}
    \begin{remark}
        According to equation \hyperref[eq:2.15]{(\ref{eq:2.15})}, the number of infected individuals at time $t$ is
    \begin{equation*}
        I(t)=e^{-(\nu+D)(t-t_{0})}I_{0}+\underset{d=1}{\overset{t-t_0}{\sum}}e^{-(\nu+D)d}N(t-d).
    \end{equation*}
	Based on equation \hyperref[eq:2.17]{(\ref{eq:2.17})}, the
	number of recovered individuals at time $t$ is 
	\begin{equation*}
		R(t)=e^{-\delta(t-t_{0})}R_{0}+\nu\underset{d_{1}=t_{0}}{\overset{
				t-1}{\sum}}e^{-\delta
			(t-d_1)}\left[e^{-(\nu+D)(d_1-t_{0})}I_{0}+\underset{d_{2}=1}{\overset{
				d_{1}-t_{0}}{\sum}}e^{-(\nu +D)d_2}N(d_1-d_2)\right].  
	\end{equation*}
    \end{remark}
    According to equation \hyperref[eq:3.22]{(\ref%
		{eq:3.22})}, in the discrete model, the expression of $\mathcal{R}(t,a)$ is
        \begin{equation*}
            \mathcal{R}(t,a)=\sum_{d=0}^{\lfloor a_+ \rfloor}\mathcal{R}(t,d).
        \end{equation*}
    Therefore, to derive the effective reproduction number $\mathcal{R}(t)$, we only need to identify the reproductive power $\mathcal{R}(t,a),\,a=0,1,...,\lfloor a_+ \rfloor$.
    \subsection*{(a) The calculation of $\mathcal{R}(t,a)$ in discrete case }
		The following analysis only considers that the Kermack-McKendrick model starts from a single cohort of infected individuals composed of $I_0$ individuals, all with an age of infection $a=0$ at time $t_0$. Then, equation \hyperref[eq:3.23]{(\ref{eq:3.23})} becomes 
		\begin{equation*} 
			N(t)=\mathcal{R}(t,t-t_{0})I_{0}+\underset{d=1}{\overset{t-t_{0}}{\sum}}%
			\mathcal{R}(t,d)N(t-d).
            \tag{5.2}
		\label{eq:5.2}
		\end{equation*}
		From \hyperref[eq:5.2]{(\ref{eq:5.2})}, we obtain
		\begin{equation*} 
			\mathcal{R}(a+t_0,a) = \frac{1}{I_0}\left[N(a+t_0)-\sum_{d=1}^{a} \mathcal{R}(a+t_0,d)  N(a+t_0-d)\right].\tag{5.3}
			\label{eq:5.3}
		\end{equation*}  
        To identify the parameter $\mathcal{R}(t,a)$ in the discrete case, we first extend the time origin of the multiple cohorts in \hyperref[eq:3.18]{(\ref{eq:3.18})} from $t_0$ to $t_0+k$.
        Assume that the initial distribution of infected individuals consists of $n_k\geq1(n_k,k\in\mathbb{Z}_{+})$ cohorts with an age of infection $a_k^{1}<a_k^{2}<...<a_k^{n_k}$ at the time $t_0+k(k>0)$. That is
		\begin{equation*}
			i_{0}^k(a)=i(t_0+k,a)=I_{k}^{1}\delta_{a_k^{1}}(a)+I_{k}^{2}\delta_{a_k^{2}}(a)+...+I_{k}^{n_k}\delta_{a_k^{n_k}}(a),
		\end{equation*}
		where $I_{k}^{j}$ is the number of infected individuals in the $j^{th}$-cohort at time $%
		t_{0}+k$. Based on equation \hyperref[eq:2.14]{(\ref*{eq:2.14})}, the flow of infected individuals $t\rightarrow N(t)$ satisfies 
		\begin{equation*}
			N(t)=\tau(t)S(t)\left[\underset{j=1}{\overset{n_k}{\sum}}\Gamma(a_k^j+(t-t_{0}-k))%
			\frac{I_{k}^{j}}{e^{-(\nu+D)a_k^{j}}}+\int_{0}^{t-t_{0}-k}\Gamma(a)N(t-a)da\right],\forall t\geq t_0+k,\tag{5.4}  \label{eq:5.4}
		\end{equation*}
		where $S(t)$ still satisfies \hyperref[eq:3.19]{(\ref*{eq:3.19})}.
        Based on equation \hyperref[eq:5.4]{(\ref*{eq:5.4})} and equation \hyperref[eq:3.21]{(\ref*{eq:3.21})}, the flow of infected individuals over time can be described by the following Volterra integral equation:
		
		\begin{equation*}
			N(t)=\underset{j=1}{\overset{n_k}{\sum}}\mathcal{R}(t,a_k^{j}+(t-t_{0}-k))%
			\frac{I_{k}^{j}}{e^{-(\nu+D)a_k^{j}}}+\int_{0}^{t-t_{0}-k}\mathcal{R}(t,a)N(t-a)da,\forall t\geq t_0+k.\tag{5.5}  \label{eq:5.5}
		\end{equation*}
        In this discrete-time framework, an outbreak that begins with a single cohort of infected individuals follows a specific temporal pattern. At time \( t_0 + k \) (where \( k \in \mathbb{N} \)), there are at most \( n_k = k + 1 \) cohorts present, with possible ages of infection denoted as \( a_k^j = j \) for \( j = 0, 1, \ldots, k \). 
        According to equation \hyperref[eq:5.5]{(\ref{eq:5.5})}, we obtain
        \begin{equation*}
            N(t)=\underset{j=0}{\overset{k}{\sum}}\mathcal{R}(t,j+(t-t_{0}-k))%
			\frac{I_{k}^{j}}{e^{-(\nu+D)j}}+\underset{d=1}{%
			\overset{t-t_0-k}{\sum}}\mathcal{R}(t,d)N(t-d),\forall t\geq t_0+k.\tag{5.6}
			\label{eq:5.6}
        \end{equation*}
        Additionally, we use \( I_k^j \) to represent the number of individuals with an infection age \( a_k^j \) at time $t_0+k$, which satisfies the following equation, 
        \begin{equation*}
         I_k^j=i(t_0+k,j),\,j = 0, 1, \ldots, k,\,k\in \mathbb{N},\, I_0^0=I_0\neq0.\tag{5.7}
			\label{eq:5.7}
        \end{equation*}In the discrete case, based on \hyperref[eq:2.13]{(\ref{eq:2.13})}, we obtain 
        \begin{equation*}
            N(t_0+k)=i(t_0+k,0)=I_k^0,
        \end{equation*}
         and 
\begin{equation*}
    i(t+1,a+1) =N(t-a)e^{-(\nu+D)(a+1)}=i(t,a)e^{-(\nu+D)},\, t=t_0+k,\,k\in\mathbb{Z}^+,\,a=0,1,2,...,\lfloor a_+ \rfloor.\tag{5.8}
			\label{eq:5.8}
\end{equation*}
        Combining equation \hyperref[eq:5.7]{(\ref{eq:5.7})} and equation \hyperref[eq:5.8]{(\ref{eq:5.8})}, we obtain
        \begin{equation*}
            \frac{I_k^j}{e^{-(\nu+D)j}}=I_{k-j}^0=N(t_0+k-j).
            \tag{5.9}
			\label{eq:5.9}
        \end{equation*}
        
       When an infectious disease begins with a single cohort of infected individuals with an age of infection $a=0$ at time $t_0$, it is not possible for individuals with an infection age of \( a \) (where \( a > 0 \)) to exist at time \( t_0 \). Therefore, analyzing the case \( \mathcal{R}(t_0, a) \) (for \( a > 0 \)) is meaningless. Furthermore, in this condition, individuals with an infection age of \( a \) (where \( a > t - t_0 \)) cannot appear at time \( t \), which makes studying the case \( \mathcal{R}(t, a) \) (for \( a > t - t_0 \)) equally meaningless.

       We begin by discussing the reproductive power $\mathcal{R}(t_0+k,a),\,0\leq k<\lfloor a_+ \rfloor + 1 ,\,0\leq a\leq k$. The analysis for the reproductive power $\mathcal{R}(t_0+k,a),\, k\geq\lfloor a_+ \rfloor + 1 ,\,0\leq a\leq k$ is similar and will be mentioned later. 
       \subsection*{(a.1)  The calculation for $\mathcal{R}(t_0+m,a),\,0\leq m<\lfloor a_+ \rfloor + 1 ,\,0\leq a\leq m$}

We present a method to obtain a linear system of equations for $\mathcal{R}(t_0+k,a),\,0\leq k<\lfloor a_+ \rfloor + 1 ,\,0\leq a\leq k$. Solving this system provides the values of \( \mathcal{R}(t_0+k,a) \) for \( a = 0, 1, \ldots, k \). First, we illustrate the process of obtaining the linear system of equations.

\textbf{Step $1$ $\mathcal{R}(t_0,0)$:} When $t=t_0$, the model starts from a single cohort of infected individuals. Substituting $t=t_0$ into equation \hyperref[eq:5.3]
{(\ref{eq:5.3})}, we obtain
\begin{equation*}
    \mathcal{R}(t_0,0)=\frac{N(t_0)}{I_0},
\end{equation*}
hence, $\mathcal{R}(t_0,0)$ exists and is unique.

\textbf{Step $2$ $\mathcal{R}(t_0+1,0),\mathcal{R}(t_0+1,1)$:} When $k=0$ and $t=t_0+1$, according to \hyperref[eq:5.6]
{(\ref{eq:5.6})}, we obtain
\begin{equation*}
    N(t_0+1)=\mathcal{R}(t_0+1,1)(I_0^0+N(t_0))
\end{equation*}
When $k=1$ and $t=t_0+1$, according to \hyperref[eq:5.6]
{(\ref{eq:5.6})}, we obtain
\begin{equation*}
    N(t_0+1)=\mathcal{R}(t_0+1,0)\frac{I_1^0}{e^{-(\nu+D)a_1^0}}+\mathcal{R}(t_0+1,1)\frac{I_1^1}{e^{-(\nu+D)a_1^1}}.
\end{equation*}
hence, if $N(t_0+1)\neq 0$, $\mathcal{R}(t_0+1,0)$ and $\mathcal{R}(t_0+1,1)$ exist and are unique.

\textbf{Step $3$ $\mathcal{R}(t_0+m,a)(a=0,1,...,m,\, 2\leq m< \lfloor a_+ \rfloor + 1)$: } 
The condition $t=t_0+m$ is satisfied in equation \hyperref[eq:5.6]
{(\ref{eq:5.6})} during this step. When $k=0$, we obtain
\begin{equation*}
    N(t_0+m)=\mathcal{R}(t_0+m,m)\frac{I_0^0}{e^{-(\nu+D)}a_0^0}+\underset{d=1}{\overset{m}{\sum}}\mathcal{R}(t_0+m,d)N(t_0+m-d).
\end{equation*}
When $k=1$, we obtain
\begin{equation*}
    N(t_0+m)=\mathcal{R}(t_0+m,m-1)\frac{I_1^0}{e^{-(\nu+D)a_1^0}}+\mathcal{R}(t_0+m,m)\frac{I_1^1}{e^{-(\nu+D)a_1^1}}+\underset{d=1}{\overset{m-1}{\sum}}\mathcal{R}(t_0+m,d)N(t_0+m-d).
\end{equation*}
When $k=s(s=2,...,m-1)$, it follows that 
\begin{equation*}
    N(t_0+m)=\underset{j=0}{\overset{s}{\sum}}\mathcal{R}(t_0+m,a_s^{j}+(m-s))\frac{I_{s}^{j}}{e^{-(\nu+D)a_s^{j}}}+\underset{d=1}{\overset{m-s}{\sum}}\mathcal{R}(t_0+m,d)N(t_0+m-d).
\end{equation*}
 When $k=m$, we obtain 
\begin{equation*}
    N(t_0+m)=\underset{j=0}{\overset{m}{\sum}}\mathcal{R}(t_0+m,a_m^{j})\frac{I_{m}^{j}}{e^{-(\nu+D)a_m^{j}}}.
\end{equation*}

Based on the above steps, we let $Co\mathcal{R}_m$ be the coefficient matrix of the linear system of equations which is\\
\begin{equation*}
    Co\mathcal{R}_m
    \scalebox{0.9}{%
  \begingroup
  \setlength{\arraycolsep}{3pt} 
  \renewcommand{\arraystretch}{0.9} 
  $
  \begin{pmatrix}
\mathcal{R}(t_0+m,0) \\
\mathcal{R}(t_0+m,1) \\
\mathcal{R}(t_0+m,2) \\
\vdots  \\
\mathcal{R}(t_0+m,m-2)\\
\mathcal{R}(t_0+m,m-1)\\
\mathcal{R}(t_0+m,m)
\end{pmatrix}= 
\begin{pmatrix}
N(t_0+m) \\
N(t_0+m) \\
N(t_0+m) \\
\vdots  \\
N(t_0+m) \\
N(t_0+m)\\
N(t_0+m)
\end{pmatrix}
  $,
  \endgroup
}
\tag{5.10}
			\label{eq:5.10}
\end{equation*}
where the matrix $Co\mathcal{R}_m$ of size $(m+1) \times (m+1),\,1\leq m<\lfloor a_+ \rfloor + 1$ is given by
\begin{equation*}
\scalebox{0.95}{%
  \begingroup
  \setlength{\arraycolsep}{3pt} 
  \renewcommand{\arraystretch}{1.0} 
  $
    \begin{pmatrix}
0 & N(t_0+m-1)  & \cdots &N(t_0+2)& N(t_0+1)&I_0^0+N(t_0) \\
0 & N(t_0+m-1)  & \cdots &N(t_0+2)& \frac{I_{1}^{0}}{e^{-(\nu+D)a_1^{0}}}+N(t_0+1)& \frac{I_{1}^{1}}{e^{-(\nu+D)a_1^{1}}} \\
\vdots & \vdots  & \ddots & \vdots  & \vdots & \vdots\\
0 & N(t_0+m-1)  & \cdots &\frac{I_{m-3}^{m-5}}{e^{-(\nu+D)a_{m-3}^{m-5}}}& \frac{I_{m-3}^{m-4}}{e^{-(\nu+D)a_{m-3}^{m-4}}}& \frac{I_{m-3}^{m-3}}{e^{-(\nu+D)a_{m-3}^{m-3}}}\\
0 & N(t_0+m-1)  & \cdots &\frac{I_{m-2}^{m-4}}{e^{-(\nu+D)a_{m-2}^{m-4}}}& \frac{I_{m-2}^{m-3}}{e^{-(\nu+D)a_{m-2}^{m-3}}}& \frac{I_{m-2}^{m-2}}{e^{-(\nu+D)a_{m-2}^{m-2}}}\\
0 & N(t_0+m-1)+\frac{I_{m-1}^{0}}{e^{-(\nu+D)a_{m-1}^{0}}} & \cdots &\frac{I_{m-1}^{m-3}}{e^{-(\nu+D)a_{m-1}^{m-3}}}& \frac{I_{m-1}^{m-2}}{e^{-(\nu+D)a_{m-1}^{m-2}}}& \frac{I_{m-1}^{m-1}}{e^{-(\nu+D)a_{m-1}^{m-1}}}\\
\frac{I_{m}^{0}}{e^{-(\nu+D)a_m^{0}}} & \frac{I_{m}^{1}}{e^{-(\nu+D)a_m^{1}}} & \cdots &\frac{I_{m}^{m-2}}{e^{-(\nu+D)a_m^{m-2}}}& \frac{I_{m}^{m-1}}{e^{-(\nu+D)a_m^{m-1}}} & \frac{I_{m}^{m}}{e^{-(\nu+D)a_m^{m}}} 
  \end{pmatrix}
  $
  \endgroup
  
}.
\end{equation*}
\begin{theorem}
\label{Th5.2} 
When $0\leq m<\lfloor a_+ \rfloor + 1$, if $\prod_{j=1}^{m}N(t_0+j),m\in\mathbb{Z}^+$, the system of linear equations in \hyperref[eq:5.10]
{(\ref{eq:5.10})} has a unique solution when $t=t_0+m$.
\end{theorem}
\begin{proof}
The case $m=0$ is obvious. We now prove the case for $1\leq m<\lfloor a_+ \rfloor + 1$.
    According to equation \hyperref[eq:5.9]
{(\ref{eq:5.9})}, the linear system of equations \hyperref[eq:5.10]
{(\ref{eq:5.10})} becomes
\begin{equation*}
\scalebox{0.95}{%
  \begingroup
  \setlength{\arraycolsep}{3pt} 
  \renewcommand{\arraystretch}{0.9} 
  $
  \begin{pmatrix}
0 & N(t_0+m-1)  & \cdots &N(t_0+2)& N(t_0+1)&I_0^0+N(t_0) \\
0 & N(t_0+m-1)  & \cdots &N(t_0+2)& 2N(t_0+1)& I_0^0 \\
\vdots & \vdots  & \ddots & \vdots& \vdots& \vdots \\
0 & N(t_0+m-1)  & \cdots &N(t_0+2)& N(t_0+1)& I_0^0\\
0 & N(t_0+m-1)  & \cdots &N(t_0+2)& N(t_0+1)& I_0^0\\
0 & 2N(t_0+m-1) & \cdots &N(t_0+2)& N(t_0+1)& I_0^0\\
N(t_0+m)& N(t_0+m-1) & \cdots &N(t_0+2)& N(t_0+1)& I_0^0
  \end{pmatrix}
  \begin{pmatrix}
\mathcal{R}(t_0+m,0) \\
\mathcal{R}(t_0+m,1) \\
\mathcal{R}(t_0+m,2) \\
\vdots  \\
\mathcal{R}(t_0+m,m-2)\\
\mathcal{R}(t_0+m,m-1)\\
\mathcal{R}(t_0+m,m)
\end{pmatrix}= 
\begin{pmatrix}
N(t_0+m) \\
N(t_0+m) \\
N(t_0+m) \\
\vdots  \\
N(t_0+m) \\
N(t_0+m)\\
N(t_0+m)
\end{pmatrix}.
  $
  \endgroup
}
\end{equation*}
The absolute value of the determinant of $Co\mathcal{R}_m$ is
\begin{align*}
|\det(Co\mathcal{R}_m)| &= N(t_0+m)\left|\det
\begin{pmatrix}
 N(t_0+m-1) & N(t_0+m-2) & \cdots &N(t_0+2)& N(t_0+1)&I_0^0+N(t_0) \\
 N(t_0+m-1) & N(t_0+m-2) & \cdots &N(t_0+2)& 2N(t_0+1)& I_0^0 \\
 \vdots & \vdots & \ddots & \vdots& \vdots& \vdots \\
 N(t_0+m-1) & N(t_0+m-2) & \cdots &N(t_0+2)& N(t_0+1)& I_0^0\\
 N(t_0+m-1) & 2N(t_0+m-2) & \cdots &N(t_0+2)& N(t_0+1)& I_0^0\\
 2N(t_0+m-1) &N(t_0+m-2) & \cdots &N(t_0+2)& N(t_0+1)& I_0^0
\end{pmatrix}\right| \\
&= N(t_0+m)\left|\det
\begin{pmatrix}
 N(t_0+m-1) & N(t_0+m-2) & \cdots &N(t_0+2)& N(t_0+1)&I_0^0+N(t_0) \\
 0 & 0 & \cdots &0& N(t_0+1)& -N(t_0) \\
 \vdots & \vdots & \ddots & \vdots& \vdots& \vdots \\
 0 & 0 & \cdots &0& 0& -N(t_0) \\
 0 & N(t_0+m-2) & \cdots &0& 0& -N(t_0) \\
 N(t_0+m-1) &0 & \cdots &0& 0& -N(t_0)
\end{pmatrix} \right|\\
&= \prod_{j=1}^{m}N(t_0+j)\left|\det
\begin{pmatrix}
1 & 1 & \cdots &1& 1&I_0^0+N(t_0) \\
 0 & 0 & \cdots &0& 1& -N(t_0) \\
 \vdots & \vdots & \ddots & \vdots& \vdots& \vdots \\
 0 & 0 & \cdots &0& 0& -N(t_0)\\
 0 & 1 & \cdots &0& 0& -N(t_0)\\
 1 &0 & \cdots &0& 0& -N(t_0)
\end{pmatrix}\right| \\
&= \prod_{j=1}^{m}N(t_0+j)
\left|\det\begin{pmatrix}
1 & 1 & \cdots &1& 1&I_0^0+(m+1)N(t_0) \\
 0 & 0 & \cdots &0& 1& 0 \\
 \vdots & \vdots & \ddots & \vdots& \vdots& \vdots \\
 0 & 0 & \cdots &0& 0& 0\\
 0 & 1 & \cdots &0& 0& 0\\
 1 &0 & \cdots &0& 0& 0
\end{pmatrix}\right| 
\end{align*}
Therefore, the existence and uniqueness of solutions to equation \hyperref[eq:5.10]{(\ref{eq:5.10})} are proved under the condition
\begin{equation*}
    \prod_{j=1}^{m}N(t_0+j)\neq 0
\end{equation*}

\end{proof}
\subsection*{(a.2)  The calculation for $\mathcal{R}(t_0+m,a),\, m\geq\lfloor a_+ \rfloor + 1 ,\,0\leq a\leq m$}
When $m\geq\lfloor a_+ \rfloor + 1$, we define $m=\lfloor a_+ \rfloor+w,w\in\mathbb{Z}^+$ and have the similar result of Theorem \ref{Th5.2}. Based on equation \hyperref[eq:5.6]
{(\ref{eq:5.6})}, we obtain
\begin{align*}
		N(t_0+\lfloor a_+ \rfloor+w) & =\underset{j=0}{\overset{k}{\sum}}\mathcal{R}(t_0+\lfloor a_+ \rfloor+w,a_k^{j}+(\lfloor a_+ \rfloor+w-k))%
			\frac{I_{k}^{j}}{e^{-(
            \nu+D)a_k^{j}}} \\
		& +\sum_{d=1}^{\lfloor a_+ \rfloor+w-k}\mathcal{R}(t_0+\lfloor a_+ \rfloor+w,d)N(t_0+\lfloor a_+ \rfloor+w-d), k=w,w+1,...,w+\lfloor a_+ \rfloor.\tag{5.11}
			\label{eq:5.11}
	\end{align*}
Based on equation \hyperref[eq:5.11]
{(\ref{eq:5.11})}, we obtain the linear system satisfied by $\mathcal{R}(t_0+\lfloor a_+ \rfloor+w,a)(a=0,1,...,\lfloor a_+ \rfloor)$ as
\begin{equation*}
    Co\mathcal{R}_m
    \scalebox{0.9}{%
  \begingroup
  \setlength{\arraycolsep}{3pt} 
  \renewcommand{\arraystretch}{0.9} 
  $
  \begin{pmatrix}
\mathcal{R}(t_0+\lfloor a_+ \rfloor+w,0) \\
\mathcal{R}(t_0+\lfloor a_+ \rfloor+w,1) \\
\mathcal{R}(t_0+\lfloor a_+ \rfloor+w,2) \\
\vdots  \\
\mathcal{R}(t_0+\lfloor a_+ \rfloor+w,\lfloor a_+ \rfloor-1)\\
\mathcal{R}(t_0+\lfloor a_+ \rfloor+w,\lfloor a_+ \rfloor)
\end{pmatrix}= 
\begin{pmatrix}
N(t_0+\lfloor a_+ \rfloor+w) \\
N(t_0+\lfloor a_+ \rfloor+w) \\
N(t_0+\lfloor a_+ \rfloor+w) \\
\vdots  \\
N(t_0+\lfloor a_+ \rfloor+w)\\
N(t_0+\lfloor a_+ \rfloor+w)
\end{pmatrix}
  $,
  \endgroup
}
\tag{5.12}
			\label{eq:5.12}
\end{equation*}
where the matrix $Co\mathcal{R}_m$ of size $(\lfloor a_+ \rfloor+1)\times(\lfloor a_+ \rfloor+1)$ is equal to
\begin{equation*}
\scalebox{0.85}{%
  \begingroup
  \setlength{\arraycolsep}{2.5pt}
  \renewcommand{\arraystretch}{0.45}
  \large 
  $
   \begin{pmatrix}
   0& N(t_0+\lfloor a_+ \rfloor+w-1) &  \cdots& N(t_0+w+1)& N(t_0+w)+\frac{I_{w}^{0}}{e^{-(\nu+D)a_{w}^{0}}}\\
   0 & N(t_0+\lfloor a_+ \rfloor+w-1) &  \cdots & N(t_0+w+1)+\frac{I_{w+1}^{0}}{e^{-(\nu+D)a_{w+1}^{0}}}& \frac{I_{w+1}^{1}}{e^{-(\nu+D)a_{w+1}^{1}}}\\
   \vdots & \vdots &  \ddots & \vdots  & \vdots\\
   0 & N(t_0+\lfloor a_+ \rfloor+w-1)  & \cdots & \frac{I_{w+\lfloor a_+ \rfloor-2}^{\lfloor a_+ \rfloor-3}}{e^{-(\nu+D)a_{w+\lfloor a_+ \rfloor-2}^{\lfloor a_+ \rfloor-3}}}& \frac{I_{w+\lfloor a_+ \rfloor-2}^{\lfloor a_+ \rfloor-2}}{e^{-(\nu+D)a_{w+\lfloor a_+ \rfloor-2}^{\lfloor a_+ \rfloor-2}}} \\
   0 & N(t_0+\lfloor a_+ \rfloor+w-1)+\frac{I_{w+\lfloor a_+ \rfloor-1}^{0}}{e^{-(\nu+D)a_{w+\lfloor a_+ \rfloor-1}^{0}}} &  \cdots & \frac{I_{w+\lfloor a_+ \rfloor-1}^{\lfloor a_+ \rfloor-2}}{e^{-(\nu+D)a_{w+\lfloor a_+ \rfloor-1}^{\lfloor a_+ \rfloor-2}}}&\frac{I_{w+\lfloor a_+ \rfloor-1}^{\lfloor a_+ \rfloor-1}}{e^{-(\nu+D)a_{w+\lfloor a_+ \rfloor-1}^{\lfloor a_+ \rfloor-1}}} \\
   \frac{I_{w+\lfloor a_+ \rfloor}^{0}}{e^{-(\nu+D)a_{w+\lfloor a_+ \rfloor}^{0}}} & \frac{I_{w+\lfloor a_+ \rfloor}^{1}}{e^{-(\nu+D)a_{w+\lfloor a_+ \rfloor}^{1}}} &  \cdots & \frac{I_{w+\lfloor a_+ \rfloor}^{\lfloor a_+ \rfloor-1}}{e^{-(\nu+D)a_{w+\lfloor a_+ \rfloor}^{\lfloor a_+ \rfloor-1}}} & \frac{I_{w+\lfloor a_+ \rfloor}^{\lfloor a_+ \rfloor}}{e^{-(\nu+D)a_{w+\lfloor a_+ \rfloor}^{\lfloor a_+ \rfloor}}}
  \end{pmatrix}.
  $
  \endgroup
}
\end{equation*}
According to equation \hyperref[eq:5.9]
{(\ref{eq:5.9})}, the coefficient matrix $Co\mathcal{R}_m$ of the linear system of equations \hyperref[eq:5.12]
{(\ref{eq:5.12})} becomes\\
\begin{equation*}
\scalebox{0.9}{
$
  \begin{pmatrix}
   0 & N(t_0+\lfloor a_+ \rfloor+w-1) &  \cdots &N(t_0+w+2)& N(t_0+w+1)& 2N(t_0+w)\\
   0 & N(t_0+\lfloor a_+ \rfloor+w-1) & \cdots &N(t_0+w+2)& 2N(t_0+w+1)& N(t_0+w)\\
   \vdots & \vdots  & \ddots & \vdots \\
   0 & N(t_0+\lfloor a_+ \rfloor+w-1) & \cdots &N(t_0+w+2)& N(t_0+w+1)& N(t_0+w) \\
   0 & 2N(t_0+\lfloor a_+ \rfloor+w-1) & \cdots &N(t_0+w+2) & N(t_0+w+1)&N(t_0+w) \\
   N(t_0+\lfloor a_+ \rfloor+w) &  N(t_0+\lfloor a_+ \rfloor+w-1) &  \cdots & N(t_0+w+2)& N(t_0+w+1) & N(t_0+w)
  \end{pmatrix}
  .
$
}
\end{equation*}
Similar to the proof of Theorem \ref{Th5.2}, the condition for the existence and uniqueness of the solution to equation \hyperref[eq:5.12]
{(\ref{eq:5.12})} is $\prod_{j=0}^{\lfloor a_+ \rfloor}N(t_0+\lfloor a_+ \rfloor+w-j)\neq0$.

	\medskip 
	\section{Application}
	\label{Section6}
    Section \ref{Section5} provides a method: In the discrete situation, if the number of newly infected individuals is known, the reproductive power can be calculated, which can be used to further derive the effective reproduction number. However, when new pathogens emerge, a lack of medical knowledge and the absence of standardized diagnostic protocols can lead to delayed detection during the initial transmission phases. This results in incomplete early data records. We demonstrate two methods applied to the data to deal with different situations.
    \subsection*{\textrm{(a)Type 1: Well-documented infection records}}

    The daily new infection records from Singapore's 2003 outbreak of Severe Acute Respiratory Syndrome (SARS) are complete enough to allow for a direct computational analysis of $\mathcal{R}(t,a)$. When working with real-world data on the number of newly infected individuals, factors such as measurement errors, reporting lags, and methodological inconsistencies can introduce significant noise and irregularities. Applying a 7-day moving average helps to reduce noise in the data and provides a clearer representation of the underlying structure and primary trends. This process provides a more reliable foundation for subsequent modeling, analysis, and decision-making.
 \begin{figure}
     \centering
     \includegraphics[width=0.7\linewidth]{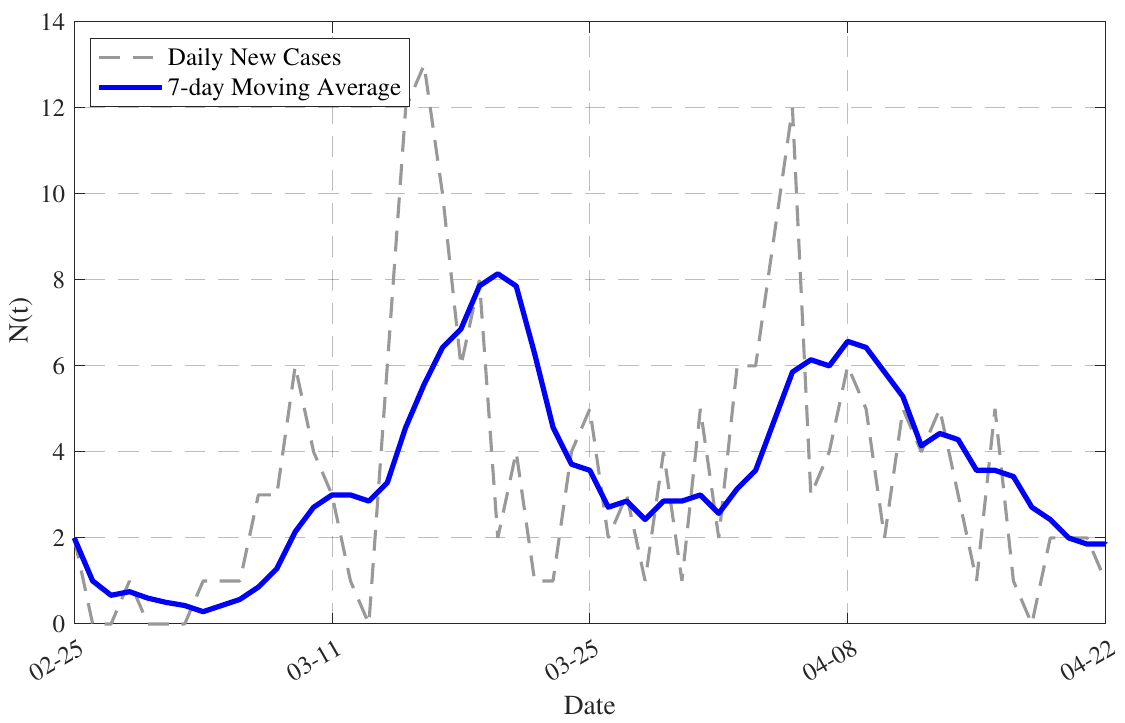}
     \caption{\textit{The gray dashed line in the figure represents the daily new infection counts reported in \cite{ref19}. Meanwhile, the blue solid line shows the corresponding time series processed using a rolling weekly average (7-day moving average).}}
        \label{fig:1}
 \end{figure}
    Using the data shown in Figure \ref{fig:1}, we calculate the reproductive power \( \mathcal{R}(t, a) \) for the infectious disease based on equations \hyperref[eq:5.10]{(\protect\ref*{eq:5.10})} and \hyperref[eq:5.12]{(\protect\ref*{eq:5.12})}. The MATLAB implementation consists of two functions in \url{https://github.com/LiJiayi-Jenny/Kermack-McKendrick-model-with-age-of-infection-and-reinfection.git}. We present the result of computed $\mathcal{R}(t,a)$ in Figure \ref{fig:2}.
   
    \begin{figure}
        \centering
        \includegraphics[width=0.7\linewidth]{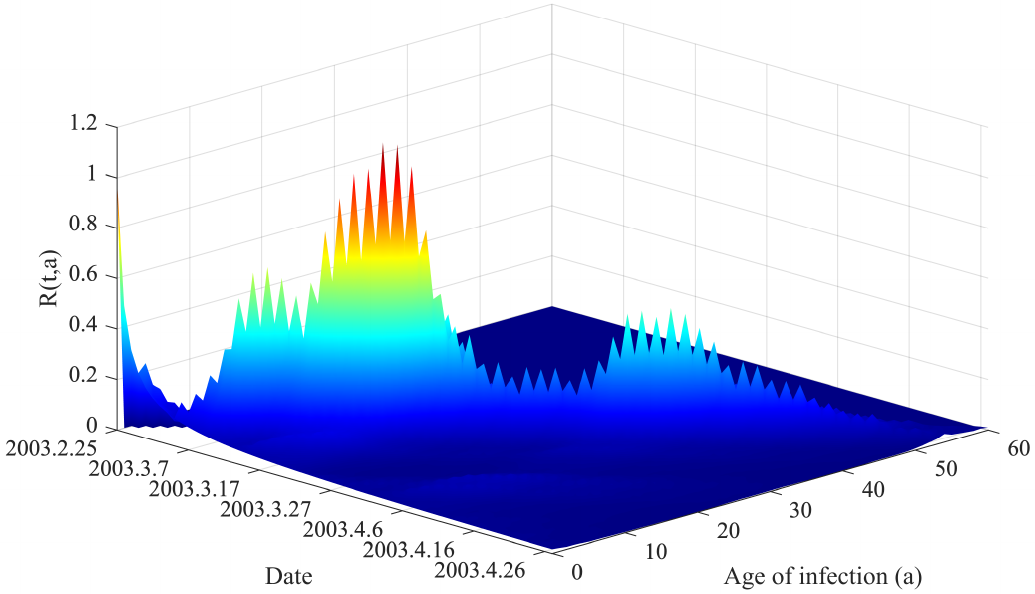}
        \caption{\textit{This figure depicts the variation of the reproductive power $\mathcal{R}(t,a)$ with respect to time $t$ and infection age $a$ from February 25, 2003, to April 26, 2003, with the parameter set to $a_+=60$ in Assumption \ref{ASS2.1}, $w=1$ in equation \hyperref[eq:5.12]{(\protect\ref*{eq:5.12})}, and $I_0=2$ in \hyperref[eq:2.7]{(\protect\ref*{eq:2.7})}.}}
        \label{fig:2}
    \end{figure}
    Fixing the parameters \( t \) and \( a \) from Figure \ref{fig:2}, we obtain a figure of cross section in Figure \ref{fig:3}.
\begin{figure}
    \centering
    \includegraphics[width=0.7\linewidth]{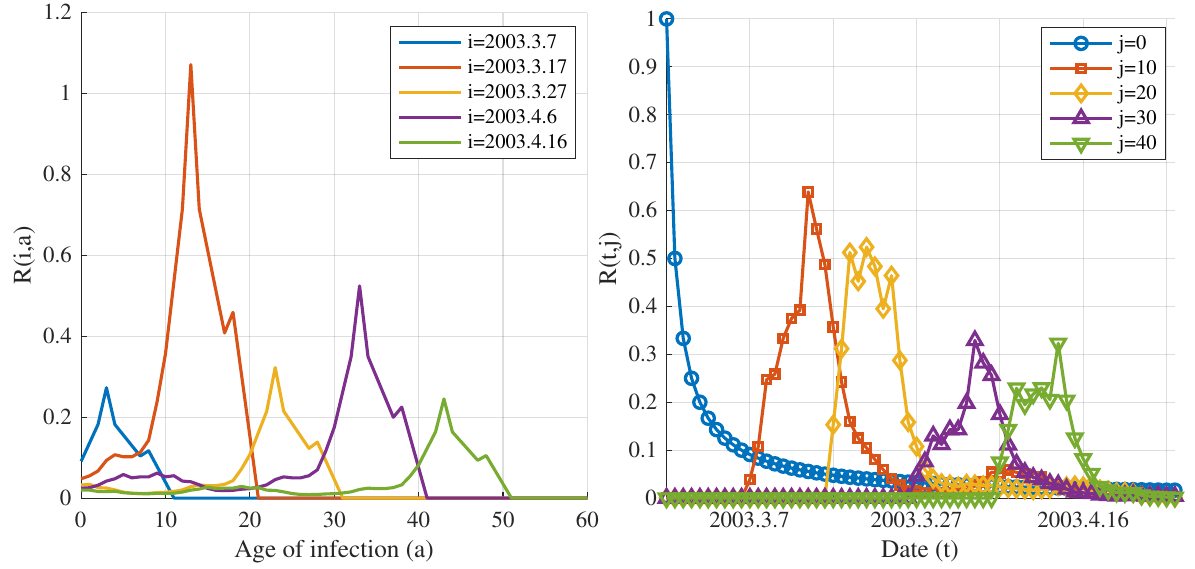}
    \caption{\textit{The left panel depicts the variation of the reproductive power $\mathcal{R}(i,a)$ as a function of infection age $a$, for five different dates $i$ at 10-day intervals. All five curves (blue, red, yellow, purple, green) exhibit a similar pattern, characterized by an initial rise to a peak followed by a brief fluctuation and decay to zero.
The right panel shows the change in the reproductive power $\mathcal{R}(t,j)$ as a function of time $t$, for five different infection ages $j$ at 10-day intervals. Here, the curves display strikingly different behaviors: four of them (red, yellow, purple, green) rise, undergo several fluctuations, and then descend, while the blue curve decreases monotonically to zero.
Collectively, these figures highlight a key difference: the time-fixed reproduction number varies with infection age in a consistent pattern, whereas the age-of-infection-fixed reproduction number evolves over time in highly distinct ways.}}
    \label{fig:3}
\end{figure}
        \subsection*{\textrm{(b) Type 2: Ambiguous infection records}}
   
    During the early stages of the outbreak, there exists documentation loss in the records, leading to considerable uncertainty in quantifying early transmission. However, these records can be used with another method. On December 31, 2019, the World Health Organization (WHO) China Country Office was notified of cases of pneumonia of unknown origin detected in Wuhan, Hubei Province, China. From December 31, 2019, to January 3, 2020, Chinese national authorities reported 44 cases of pneumonia with undetermined causes to the WHO. During this reporting period, the cause remained unidentified. WHO's official data records began on January 20, 2020. Only sporadic reports from Wuhan were available between December 31, 2019, and January 19, 2020. Facing ambiguous data records, we can utilize the estimator \(CD'(t)\) from \hyperref[eq:4.1]{(\protect\ref*{eq:4.1})}. We estimate the cumulative number of deaths provided by the WHO \cite{ref53} and the Chinese government website \cite{ref54} in Figure \ref{fig:4}, starting from January 20, 2020.
    \begin{figure}
        \centering
        \includegraphics[width=0.7\linewidth]{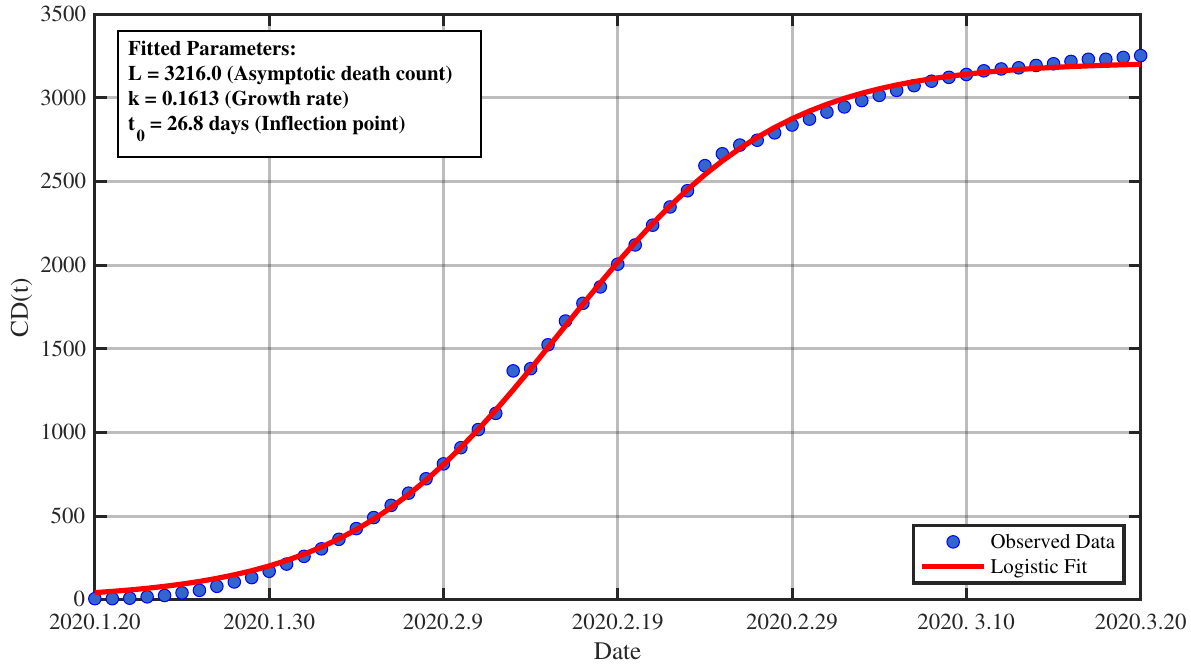}
        \caption{\textit{The fitting figure of function $CD(t)$:
        Blue points represent reported mortality data from \cite{ref53}. 
        The logistic function $CD(t) = \frac{L}{1 + e^{-k(t - t_0)}}$ 
        provides optimal fit $($red curve$)$. Fitted parameters: 
        Asymptotic death count $(L)$: $3216.0$; Growth rate $(k)$: $0.1613$; Inflection point $(t_0)$: $26.8$ days; $R^2$: $0.9993$.}
}
        \label{fig:4}
    \end{figure}
    The model in \hyperref[eq:2.2]{(\protect\ref*{eq:2.2})} assumes that the recovery rate $\nu$ and the mortality rate $D$ are constants. This assumption enables us to analyze the initial stage of the COVID-19 outbreak in China, spanning from December 31, 2019, to February 12, 2020. Based on \hyperref[eq:4.4]{(\protect\ref*{eq:4.4})} and \hyperref[eq:4.6]{(\protect\ref*{eq:4.6})}, We derive $N(t)$ and $I_0$ through the calculation of $CD'(t)$ and $CD''(t)$. During the early stage of the epidemic in China, mortality was primarily concentrated in Hubei Province. Based on \cite{ref57}, we adopt a mortality rate of $\nu = 0.88$ and a recovery rate of $D = 0.04$. Using equation \hyperref[eq:4.6]{(\protect\ref*{eq:4.6})}, we obtain the following expression based on the fitted $CD'(t)$ below.
    \begin{equation*}
        N(t)=\frac{Lke^{-k(t-t_0)}[(\nu+D-k)+(\nu+D+k)e^{-k(t-t_0)}]}{D(1+e^{-k(t-t_0)})^3}.\tag{8.1}  \label{eq:8.1}
    \end{equation*}
     We first present $N(t)$ and $CD(t)$ from December 31, 2019 in figure \ref{fig:5}.
    \begin{figure}
        \centering
        \includegraphics[width=0.8\linewidth]{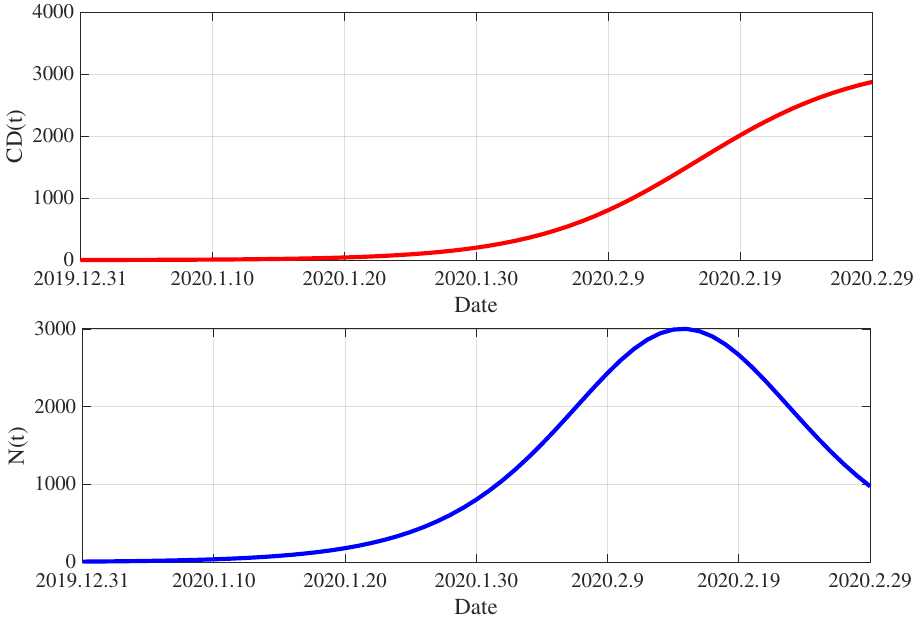}
        \caption{\textit{This figure illustrates the cumulative number of deaths $CD(t)=\frac{L}{1 + e^{-k(t - t_0)}}$ and the number of newly infected individuals $N(t)$ since December 31, 2020. The upper panel presents $CD(t)$, while the lower panel shows $N(t)$ in \hyperref[eq:8.1]{(\ref{eq:8.1})}.}}
        \label{fig:5}
    \end{figure}
Using the values of \( N(t) \) in Figure \ref{fig:5}, we calculate the reproduction number \( \mathcal{R}(t, a) \) from December 31, 2019, to February 29, 2020, using equation \hyperref[eq:5.10]{(\ref{eq:5.10})} and \hyperref[eq:5.12]{(\ref{eq:5.12})}.

\begin{figure}
    \centering
    \includegraphics[width=0.7\linewidth]{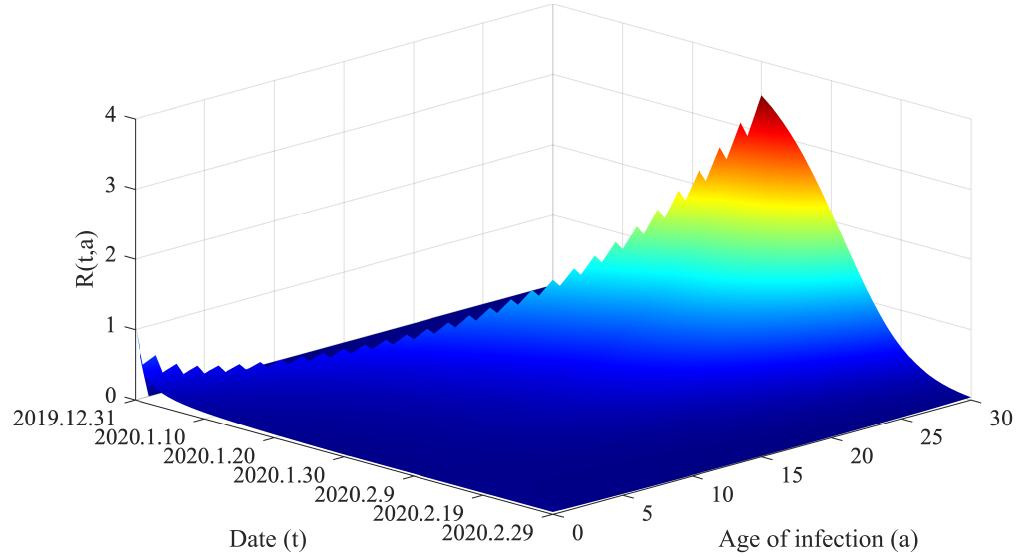}
    \caption{\textit{The figure illustrates the reproductive power \( \mathcal{R}(t, a) \). We set the observed duration \( a_+ = 30 \) days in Assumption \ref{ASS2.1}, with parameters \( \nu = 0.88 \), \( D = 0.04 \) based on \cite{ref57} , $ w=50$ in equation \hyperref[eq:5.12]{(\protect\ref*{eq:5.12})}, and \( I_0 = 6.8236 \) calculated by \hyperref[eq:4.4]{(\ref{eq:4.4})}.}}
    \label{fig:6}
\end{figure}
In this model, we assume that the recovery rate and mortality rate are constants. However, according to the records of the 2019 COVID-19 epidemic data in China provided in reference \cite{ref56}, it is only reasonable to treat the recovery and mortality rates as constants during the initial stage of the outbreak. As the epidemic evolves, these rates undergo significant changes. Therefore, we set the parameter $a_+ = 30$. Fixing the parameters \( t \) and \( a \) respectively from Figure \ref{fig:6}, we obtain a figure of cross section in Figure \ref{fig:7}.
\begin{figure}
    \centering
    \includegraphics[width=0.7\linewidth]{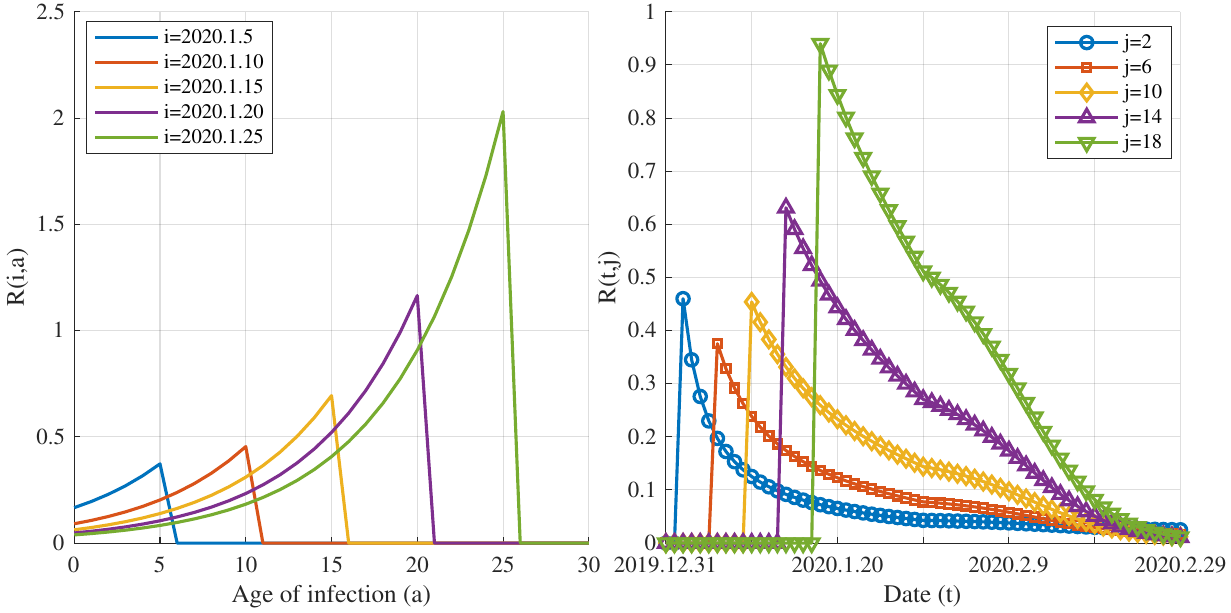}
    \caption{\textit{The left figure shows the variation of the reproductive power $\mathcal{R}(i,a)$ across different ages of infection at $5$ different dates $i$. Each curve, corresponding to a different date, initially rises, reaches a peak, and then declines to zero. The peak of each curve shifts to the right. This suggests that the reproductive number increases with the age of infection up to a certain point, after which it decreases. The peak reproductive number and the age at which it occurs both increase over time. The right figure displays the change in the reproductive power $\mathcal{R}(t,j)$ over time for five different ages of infection $j$. Each curve starts at a different level and decreases over time. The curves corresponding to higher ages of infection (e.g., \( j=18 \)) start higher and decrease faster compared to those with lower ages of infection (e.g., \( j=2 \)). This indicates that the reproductive number decreases over time for all ages of infection, but the rate and initial level of decrease vary depending on the age of infection. Both graphs show a dynamic relationship between the reproductive number and either the age of infection or the date, highlighting the changing nature of infectious disease spread over time and with infection progression.}}
    \label{fig:7}
\end{figure}

\section{Discussion}
\label{Section7}
Based on the Kermack-McKendrick model with age of infection and reinfection, this paper estimates the reproductive power $\mathcal{R}(t, a)$ from reported data, which is dependent on both time and age of infection. We establish the existence and uniqueness of solutions for $N(t)$ through the partial differential equation and the Volterra integral equation. Then, we propose a method for obtaining discrete estimates of $\mathcal{R}(t, a)$ from $N(t)$. Meanwhile, we have derived the expression for $\mathcal{R}(t,a)$ based on the cumulative number of deaths $CD(t)$ and parameters in model. This work provides a quantitative tool for long-term epidemic monitoring and precise control measures.

Regarding the key initial condition $i_0(a)$ of the model, this paper adopts the discussion on both single and multiple cohorts from reference \cite{ref2}. Under the condition that the model starts from a single cohort of infected individuals, the identification method for parameter $\mathcal{R}(t,a)$ can be further established in the discrete situation. The application in Section \ref{Section6} demonstrates the estimation results of $\mathcal{R}(t, a)$. Furthermore, the two-parameter reproduction number allows us to examine disease transmission along both the time and infection-age dimensions. As shown in Figure \ref{fig:3} and Figure \ref{fig:7}, it simultaneously displays the transmission potential among individuals of different infection ages at a fixed time point, as well as the temporal variation in transmission potential for individuals with a specific infection age.

To enhance the model's precision, it is vital to explore the form of $\beta(a)$ of different infectious diseases. Using our model with the data in Section \ref{Section4}, equation \hyperref[eq:4.10]{(\protect\ref*{eq:4.10})} shows that a deterministic expression for $\beta(a)$ enables us to obtain the expression for $\tau(t)$. Gaining a deeper biological understanding of $\beta(a)$ is crucial for improving the model’s accuracy. Current biological research primarily examines trends in pathogen load, particularly pathogen shedding. However, due to diverse host conditions, viral load (e.g., viral RNA levels) can only serve as a proxy for $\beta(a)$, which represents infectious virus \cite{ref33}. Therefore, while we can infer trends in viral dynamics, the intricate and precise relationship between viral load and $\beta(a)$ demands further investigation. Secondly, as indicated by the theorem, when $N(t)=0$ for $t=t_0+k$, with $k=1,2,\ldots$, the non-uniqueness of $\mathcal{R}(t,a)$ presents a significant challenge. How should we navigate this scenario? To address the occurrence of zero values in the discrete model, we identify several critical possibilities:
    \begin{enumerate}
        \item Missing data (unreported cases),
        \item Viral evolution influenced by herd immunity (where the virus creates variants, leading to substantial changes in its characteristics and resulting in unstable model parameters),
        \item Disease extinction caused by interventions such as isolation.
    \end{enumerate}
    The first and third cases can be effectively managed within the current model framework (e.g., employing data preprocessing techniques for missing data). However, the second case requires a more sophisticated approach, reflecting the complexities of viral evolution. Finally, our study currently omits considerations of age-of-infection-dependent recovery rates and latency periods. Future research can delve deeply into these aspects, as they may significantly influence model outcomes and real-world applications.
\section*{Acknowledgments}
This work is supported by the National Natural Science Foundation of China (No. 12271044, 12071297).
\section*{Conflict of interest}

The authors declare there is no conflict of interest.

\newpage
\section*{Appendix}
\appendix
\renewcommand{\thesection}{\Alph{section}} 
\section{The prove of Lemma \ref{lem4.2}}\label{app:A}
\begin{proof}
    Let $\varepsilon>0$. We observe that
		\begin{align*}
        \Lambda_{\kappa}(t)-I_{0}e^{-(\nu+D)(t-t_{0})}\beta(t-t_0)& =
			 I_{0}e^{-(\nu+D)(t-t_{0})}\int_{0}^{+\infty}[\beta(a+(t-t_{0}))-\beta(t-t_0)]\kappa
			p(\kappa a)da\\
			&  =I_{0}e^{-(\nu+D)(t-t_{0})}\int_{0}^{\eta}\left[\beta(a+(t-t_{0}%
			))-\beta(t-t_{0})\right]\kappa p(\kappa a)da\\
			& \,\,\,\,\,\,\, +I_{0}e^{-(\nu+D)(t-t_{0})}\int_{\eta}^{+\infty}\left[\beta(a+(t-t_{0}))-\beta(t-t_{0})\right]\kappa p(\kappa a)da.
		\end{align*}
		There exists $t_{1}>t_{0}$ such that 
		
		\[
		I_{0}e^{-(\nu+D)(t-t_{0})}\underset{a\geq0}{\sup}\beta(a)\leq
		\frac{\varepsilon}{4},\,\forall t\in[t_{1},+\infty).
		\]
		Let $\eta>0$ be such that %
		
		\[
		a\leq\eta\Rightarrow|\beta(a+t)-\beta(t)|\leq\frac{\varepsilon}{2I_0},\forall
		t\in\lbrack t_{0},t_{1}].
		\]
		Then we have %
         \begin{equation*}
    \begin{aligned}
        &\left|I_{0}e^{-(\nu+D)(t-t_{0})}\int_{0}^{\eta}\left[\beta(a+(t-t_{0}%
			))-\beta(t-t_{0})\right]\kappa p(\kappa a)da\right| \\
        &\leq \left\{  
            \begin{aligned}
        &I_{0}e^{-(\nu+D)(t-t_{0})}2\sup_{a\geq0}\beta(a), && \text{if } t \in [t_{1},+\infty), \\
        &I_{0}e^{-(\nu+D)(t-t_{0})}\int_{0}^{\eta}\frac{\varepsilon}{2}\kappa p(\kappa a)da, && \text{if } t \in [t_{0},t_{1}],
         \end{aligned}
        \right. 
    \end{aligned}
\end{equation*}
		therefore,%
		\begin{equation*}
			\left|\Lambda_{\kappa}(t)-I_{0}e^{-(\nu+D)(t-t_{0})}\beta(t-t_{0})\right|\leq\frac{\varepsilon}{2}+\left|I_{0}e^{-(\nu+D)(t-t_{0})}\int_{\eta}^{\infty
			}\left[\beta(a+(t-t_{0}))-\beta(t-t_{0})\right]\kappa p(\kappa a)da\right|.
		\end{equation*}
		According to
		\begin{equation*}
			\left|I_{0}e^{-(\nu+D)(t-t_{0})}\int_{\eta}^{\infty
			}\left[\beta(a+(t-t_{0}))-\beta(t-t_{0})\right]\kappa p(\kappa a)da\right|\leq2I_{0}e^{-(\nu+D)(t-t_{0})}\underset{a\geq0}{\sup}\beta(a)\left(1-\int_{0
			}^{\eta}\kappa p(\kappa a)da\right),
		\end{equation*}
		and,%
		
		\begin{align*}
			&  2I_{0}e^{-(\nu+D)(t-t_{0})}\underset{a\geq0}{\sup}\beta(a)\left(1-\int_{0
			}^{\eta}\kappa p(\kappa a)da\right)=2I_{0}e^{-(\nu+D)(t-t_{0})}\underset{a\geq0}{\sup}\beta(a)\left(1-\int_{0
			}^{\kappa\eta}p(a)da\right)\rightarrow0,as\text{ }\kappa\rightarrow\infty,
		\end{align*}
	    we obtain,
        \begin{align*}
    &\left|I_{0}e^{-(\nu+D)(t-t_{0})}\int_{\eta}^{\infty
			}\left[\beta(a+(t-t_{0}))-\beta(t-t_{0})\right]\kappa p(\kappa a)da\right|\\
        &\leq2I_{0}e^{-(\nu+D)(t-t_{0})}\underset{a\geq0}{\sup}\beta(a)\left(1-\int_{0
			}^{\eta}\kappa p(\kappa a)da\right)\rightarrow0,as\text{ }\kappa\rightarrow\infty.
\end{align*}
	
\end{proof}
\section{If \( N_\kappa(t) \) exists in $\Omega
$, \( N_{\kappa }(t) \) forms a fundamental sequence of functions. Therefore, there exists a function \( A(t) \) such that 
\( \underset{\kappa \rightarrow \infty }{\lim } N_{\kappa }(t) := A(t) \). }\label{app:B}
\begin{proof}
        Appendix \ref{app:B} only needs to be proven to hold on a bounded interval $[t_0,t_0+T]$. Other bounded closed intervals in $[t_0,+\infty)$ can be similarly proven.
		Firstly, find the norm of the difference between $S_{\kappa_{1}}(t)$ and $S_{\kappa_{2}}(t)(\kappa_1,\kappa_2\in \mathbb{N})$, which means
		\begin{align*}
			S_{\kappa}:=\left\Vert S_{\kappa_{1}}(t)-S_{\kappa_{2}}(t)\right\Vert  &  =||-\int_{t_{0}}^{t}%
			N_{\kappa _{1}}(m)-\delta\{\nu\int_{t_{0}}^{m}e^{-\delta(m-\tau)}[\int_{0}^{\tau-t_{0}%
			}e^{-(\nu+D)a}N_{\kappa _{1}}(\tau-a)da]d\tau\}dm\\
			&  +\int_{t_{0}}^{t}N_{\kappa _{2}}(m)-\delta\{\nu\int_{t_{0}}^{m}e^{-\delta(m-\tau
				)}[\int_{0}^{\tau-t_{0}}e^{-(\nu+D)a}N_{\kappa _{2}}(\tau-a)da]d\tau\}dm||\\
			& \leq(1+\nu\delta T^{2})T\left\Vert N_{\kappa _{1}}-N_{\kappa _{2}}\right\Vert.
		\end{align*}
By Assumption \ref{ASS2.1}, there exists a constant $\tau_{max}>0$ such that 
        \begin{equation*}
            ||\tau(t)||\leq \tau_{max}.
        \end{equation*}
		According to equation \hyperref[eq:3.6]{(\ref*{eq:3.6})}, $\forall\varepsilon>0$, there exists an integer $K>0$, $\forall\,\,\kappa_{1}, \kappa_{2}>K$ such that
		\begin{equation*}
			\Lambda:=\left\Vert \tau (t)S_{\kappa _{1}}(t)\Lambda _{\kappa _{1}}(t)-\tau
			(t)S_{\kappa _{2}}(t)\Lambda _{\kappa _{2}}(t)\right\Vert <\varepsilon+I_0\tau_{max}(1+\nu\delta T^{2})T\left\Vert N_{\kappa _{1}}-N_{\kappa _{2}}\right\Vert .
		\end{equation*}
        There exists $\alpha>0$ such that whenever $T<\alpha$ holds
		\begin{align*}
			\left\Vert N_{\kappa _{1}}(t)-N_{\kappa _{2}}(t)\right\Vert 
			& \leq \Lambda+\tau_{max}S_{\kappa}\left\Vert \int_{0}^{t-t_{0}}\beta (a)e^{-(\nu +D)a}N_{\kappa _{2}}(t-a)da\right\Vert\\&+\tau_{max}\left\Vert S_{\kappa _{1}}(t)\int_{0}^{t-t_{0}}\beta (a)e^{-(\nu
				+D)a}\left[N_{\kappa _{1}}(t-a)-N_{\kappa _{2}}(t-a)\right]da\right\Vert
			\\
			& \leq m\left\Vert N_{\kappa _{1}}(t)-N_{\kappa _{2}}(t)\right\Vert
			+\varepsilon, 0\leq m<1,
		\end{align*}
		equally, $N_{\kappa }(t)$ is a fundamental sequence of functions in $\Omega$.
	\end{proof}

\section{$\underset{\kappa \rightarrow \infty }{\lim }N_{\kappa }(t)$ is
	equal to $A(t)$ in existence, then the limit value $A(t)$ satisfies equation \hyperref[eq:3.13]{(\ref*{eq:3.13})} in every closed and bounded interval of $[t_0,+\infty)$
    .}\label{app:C}
	\begin{proof}
		Based on Assumption \ref{ASS2.1}, equation \hyperref[eq:3.6]{(\ref*{eq:3.6})}, the relationship between the Riemann and Lebesgue integrals, and the dominated convergence theorem, we obtain
		\begin{equation*}
			\underset{\kappa\rightarrow\infty}{\lim}N_{\kappa}(t)=\tau(t)\underset{%
				\kappa \rightarrow\infty}{\lim}S_{\kappa}(t)[\underset{\kappa\rightarrow
				\infty}{\lim}\Lambda_{\kappa}(t)+\int_{0}^{t-t_{0}}\beta(a)e^{-(\nu +D)a}
			\underset{\kappa\rightarrow\infty}{\lim}N_{\kappa}(t-a)da],
		\end{equation*}
		and
		\begin{align*}
			\underset{\kappa \rightarrow \infty }{\lim }S_{\kappa }(t)& =\underset{%
				\kappa \rightarrow \infty }{\lim }F(N_{\kappa }(t)) \\
			& =S_{0}-\int_{t_{0}}^{t}\underset{\kappa \rightarrow \infty }{\lim }%
			N_{\kappa }(m)-\delta \{e^{-\delta (m-t_{0})}R_{0}+\nu
			\int_{t_{0}}^{m}e^{-\delta (m-\tau )}[e^{-(\nu +D)(\tau -t_{0})}I_{0} \\
			&+\int_{0}^{\tau -t_{0}}e^{-(\nu +D)a}\underset{\kappa \rightarrow \infty }{%
				\lim }N_{\kappa }(\tau -a)da]d\tau \}dm \\
			& =F(\underset{\kappa \rightarrow \infty }{\lim }N_{\kappa }(t))=F(A(t)).
		\end{align*}
		Hence, $A(t)$ satisfies
		\begin{equation*}
			A(t)=\tau(t)F(A(t))[I_{0}\Gamma(t-t_{0})+\int_{0}^{t-t_{0}}\Gamma
			(a)A(t-a)da] .
		\end{equation*}
	\end{proof}


\newpage
	\begin{thebibliography}{99}
\bibitem[1]{ref46}
R. M. Anderson and R. M. May,
\emph{Infectious diseases of humans: Dynamics and control},
Oxford University Press, Oxford (1991).\

\bibitem[2]{ref36}
A. Assiri \emph{et al.},
\emph{Hospital outbreak of {Middle East} respiratory syndrome coronavirus},
\textbf{N. Engl. J. Med.} \textbf{369}, no.~5, 407–416 (2013). 

\bibitem[3]{ref19}
Centers for Disease Control and Prevention ({CDC}),
\emph{Severe acute respiratory syndrome---{Singapore}, 2003},
\textbf{MMWR Morb. Mortal. Wkly. Rep.} \textbf{52}, no.~18, 405–411 (2005).

\bibitem[4]{ref23}
F. Clément, B. Laroche, and F. Robin,
\emph{Analysis and numerical simulation of an inverse problem for a structured cell population dynamics model},
\textbf{Math. Biosci. Eng.} \textbf{16}, no.~4, 3018–3046 (2019). 

\bibitem[5]{ref16}
A. Crocker and D. Strömborn,
\emph{Susceptible-infected-susceptible type {COVID}-19 spread with collective effects},
\textbf{Sci. Rep.} \textbf{13}, 22600 (2023). 

\bibitem[6]{ref2}
J. Demongeot, Q. Griette, Y. Maday, and P. Magal,
\emph{A {Kermack--McKendrick} model with age of infection starting from a single or multiple cohorts of infected patients},
\textbf{Proc. R. Soc. A} \textbf{479}, Paper No. 20220678 (2022). 

\bibitem[7]{ref21}
O. Diekmann and J. A. P. Heesterbeek,
\emph{Mathematical epidemiology of infectious diseases: Model building, analysis and interpretation},
Wiley Series in Mathematical and Computational Biology, John Wiley \& Sons, Ltd., Chichester (2000).

\bibitem[8]{ref60}
Diekmann, O., Heesterbeek, H., and Britton, T. \emph{Mathematical Tools for Understanding Infectious Disease Dynamics}. \textbf{Princeton University Press}(2013). 

\bibitem[9]{ref5}
E. Dong \emph{et al.},
\emph{The {Johns Hopkins University} Center for Systems Science and Engineering {COVID}-19 Dashboard: data collection process, challenges faced, and lessons learned},
\textbf{Lancet Infect. Dis.} \textbf{22}, no.~12, e370--e376 (2022). 

\bibitem[10]{ref51}
C. Fraser,
\emph{Estimating individual and household reproduction numbers in an emerging epidemic},
\textbf{PLoS ONE} \textbf{2}, no.~8, e758 (2007). 

\bibitem[11]{ref37}
E. R. Gaunt, A. Hardie, E. C. Claas, P. Simmonds, and K. E. Templeton,
\emph{Epidemiology and clinical presentations of the four human coronaviruses 229E, {HKU1}, {NL63}, and {OC43} detected over 3 years using a novel multiplex real-time {PCR} method},
\textbf{J. Clin. Microbiol.} \textbf{48}, no.~8, 2940–2947 (2010). 

\bibitem[12]{ref27}
P. Gabriel,
\emph{Measure solutions to the conservative renewal equation},
\textbf{ESAIM Proc. Surveys} \textbf{62}, 68–78 (2018). 

\bibitem[13]{ref17}
W. S. Hart \emph{et al.},
\emph{Analysis of the risk and pre-emptive control of viral outbreaks accounting for within-host dynamics: {SARS-CoV-2} as a case study},
\textbf{Proc. Natl. Acad. Sci. U.S.A.} \textbf{120}, no.~41, e2305451120 (2023). 

\bibitem[14]{ref43}
W. O. Kermack and A. G. McKendrick,
\emph{A contribution to the mathematical theory of epidemics},
\textbf{Proc. R. Soc. Lond. A} \textbf{115}, 700–721 (1927).

\bibitem[15]{ref3}
W. O. Kermack and A. G. McKendrick,
\emph{Contributions to the mathematical theory of epidemics. {II}.---The problem of endemicity},
\textbf{Proc. R. Soc. Lond. A} \textbf{138}, 55–83 (1932).

\bibitem[16]{ref29}
E. Kuhl,
\emph{Computational epidemiology: Data-driven modeling of {COVID}-19},
Springer, Cham, (2021).

\bibitem[17]{ref50}
P. Magal and C. C. McCluskey,
\emph{Two-group infection age model including an application to nosocomial infection},
\textbf{SIAM J. Appl. Math.} \textbf{73}, no.~2, 1058–1095 (2013). 

\bibitem[18]{ref45}
A. G. McKendrick,
\emph{Applications of mathematics to medical problems},
\textbf{Proc. Edinb. Math. Soc.} \textbf{44}, 98–130 (1926).

\bibitem[19]{ref54}
National Health Commission of the People's Republic of China,
\emph{Update on the novel coronavirus pneumonia outbreak ({January} 22, 2020)} (2020).

\bibitem[20]{ref25}
H. Nishiura and G. Chowell,
\emph{The effective reproduction number as a prelude to statistical estimation of time-dependent epidemic trends},
In \emph{Mathematical and statistical estimation approaches in epidemiology}, Springer, Dordrecht, pp. 103–121 (2009).

\bibitem[21]{ref33}
O. Puhach, B. Meyer, and I. Eckerle,
\emph{{SARS-CoV-2} viral load and shedding kinetics},
\textbf{Nat. Rev. Microbiol.} \textbf{21}, no.~3, 147–161 (2023). 

\bibitem[22]{ref15}
K. J. Sharkey, C. Fernandez, K. L. Morgan, and S. E. F. Spencer,
\emph{Pair-level approximations to the spatio-temporal dynamics of epidemics on asymmetric contact networks},
\textbf{J. Math. Biol.} \textbf{53}, no.~1, 61–85 (2006). 

\bibitem[23]{ref49}
J. A. Simpson, L. Aarons, W. E. Collins, G. M. Jeffery, and N. J. White,
\emph{Population dynamics of untreated \emph{Plasmodium falciparum} malaria within the adult human host during the expansion phase of the infection},
\textbf{Parasitology} \textbf{124}, no.~3, 247–263 (2002).

\bibitem[24]{ref55}
S. Soubeyrand, J. Demongeot, and L. Roques,
\emph{Towards unified and real-time analyses of outbreaks at country-level during pandemics},
\textbf{One Health} \textbf{11}, 100187 (2020). 

\bibitem[25]{ref57}
ScienceNet, National Health Commission of the People's Republic of China, Latest news on the COVID-19 epidemic as of 24:00 on March 20, Mar. 21 (2020). 
\bibitem[26]{ref56}
The Paper,
\emph{Cumulative number of cured and discharged patients exceeds number of existing confirmed cases: When will the turning point arrive?}
Mar. 1 (2020).

\bibitem[27]{ref48}
M. J. Wawer \emph{et al.},
\emph{Rates of {HIV-1} transmission per coital act, by stage of {HIV-1} infection, in {Rakai}, {Uganda}},
\textbf{J. Infect. Dis.} \textbf{191}, no.~9, 1403–1409 (2005). 

\bibitem[28]{ref1}
G. F. Webb,
\emph{Theory of nonlinear age-dependent population dynamics},
Monographs and Textbooks in Pure and Applied Mathematics, vol. 89, Marcel Dekker, Inc., New York (1985).

\bibitem[29]{ref53}
World Health Organization,
\emph{{WHO COVID-19} dashboard},
World Health Organization, Geneva (2020).

\bibitem[30]{ref47}
X.-J. Xu, S.-J. He, and L.-J. Zhang,
\emph{Improved estimation of the effective reproduction number with heterogeneous transmission rates and reporting delays},
\textbf{Sci. Rep.} \textbf{14}, 28125 (2024). 

\bibitem[31]{ref35}
I. T. S. Yu \emph{et al.},
\emph{Evidence of airborne transmission of the severe acute respiratory syndrome virus},
\textbf{N. Engl. J. Med.} \textbf{350}, no.~17, 1731–1739 (2004). 

\bibitem[32]{ref58}
Hoskins, R. F., and Michael P. Lamoureux. \emph{Delta Functions: Introduction to Generalised Functions. 2nd ed}, \textbf{World Scientific}, pp. 26-38 (2009).

\bibitem[33]{ref59}
Wazwaz, AM. \emph{Nonlinear Volterra Integral Equations. In: Linear and Nonlinear Integral Equations}. \textbf{Springer, Berlin, Heidelberg } (2011).


\end{thebibliography}
\end{document}